\documentclass[aps,prx,twocolumn,superscriptaddress]{revtex4-2}
\usepackage{tikz}
\usetikzlibrary{calc}
\usepackage[english]{babel}
\usepackage[utf8]{inputenc}
\usepackage{indentfirst}
\usepackage{amsmath}
\usepackage{amssymb}
\usepackage{amsthm}
\usepackage{thm-restate}
\usepackage{thmtools}
\usepackage{eufrak}
\usepackage{physics}
\usepackage{bbold}
\usepackage{cancel}

\newcommand{\ZZ}{\mathcal{Z}}

\newcommand{\refst}{\ket{\emptyset}}

\newcommand{\AAA}{\mathcal{A}}

\newcommand{\NN}{\mathcal{N}}

\newcommand{\complex}{\mathbb{C}}

\newcommand{\valos}{\mathbb{R}}
\newcommand{\egesz}{\mathbb{N}}
\newcommand{\Zk}{\mathbb{Z}_2}
\newcommand{\eps}{\varepsilon}

\newcommand{\ordo}{\mathcal{O}}

\usepackage{amsthm}
\newtheorem{thm}{Theorem}

\newtheorem{conje}{Conjecture}

\newcommand{\llangle}{\langle\!\langle}
\newcommand{\rrangle}{\rangle\!\rangle}
\newcommand{\HS}[2]{{\llangle #1 | #2 \rrangle}}

\newcommand{\ffda}{\mathcal{A}_{FFD}}
\newcommand{\tffda}{\tilde{\mathcal{A}}_{FFD}}

\newtheorem*{theorem*}{Theorem}

\usepackage{ifpdf}

\ifpdf
\usepackage{epstopdf}
\usepackage[pdftex,colorlinks,urlcolor=blue,citecolor=blue,linkcolor=blue]{hyperref}
\else
\usepackage[hypertex,colorlinks,urlcolor=blue,citecolor=blue,linkcolor=blue]{hyperref}
\fi
\newcommand{\elte}{\affiliation{MTA-ELTE “Momentum” Integrable Quantum Dynamics Research Group,\\
    ELTE E\"otv\"os Lor\'and University, Budapest, Hungary}}
\newcommand{\wigner}{\affiliation{ Holographic Quantum Field Theory Research Group,\\
    HUN-REN Wigner Research Centre for Physics, Budapest, Hungary}}
\newcommand{\tokyo}{\affiliation{Department of Physics, Graduate School of Science, The University of Tokyo, \\ 7-3-1, Hongo, Bunkyo-ku, Tokyo, 113-0033, Japan}}
\begin{document}

\title{
Solving models with generalized free fermions I: Algebras and eigenstates 
}
\author{Kohei Fukai}
\tokyo
\author{Bal\'azs Pozsgay}
\elte
\author{Istv\'an Vona}
\elte\wigner

\begin{abstract}
We study quantum spin chains solvable via hidden free fermionic structures. We study the 
algebras behind such models, establishing connections to the mathematical literature of the so-called ``graph-Clifford''
or ``quasi-Clifford'' algebras. We also introduce the ``defining representation'' for such algebras, and show that this
representation actually coincides with the terms of the Hamiltonian in two relevant models:
the XY model and the ``free fermions in disguise'' model of Fendley. Afterwards we study a particular anti-symmetric
combination of commuting Hamiltonians; this is performed in a model independent way. We show that for this combination
there exists a 
reference state, and few body eigenstates can be created by the fermionic operators. Concrete application is presented
in the case of the  ``free fermions in disguise'' model.
\end{abstract}

\maketitle

\section{Introduction}

There are numerous lattice models which appear to be strongly interacting, nevertheless they admit
a complete solution using free fermions. Such models belong to the simplest many-body systems, and their exact
solvability is invaluable for theoretical physics. Famous examples of such models include the two dimensional
statistical physical Ising model (whose transfer matrix can be diagonalized using free fermions
\cite{Schulz-Mattis-Lieb}), the XX and XY quantum spin chains \cite{XX-original}, Kitaev's chain \cite{kitaev-chain} or
Kitaev's honeycomb model \cite{kitaev-honeycomb}.

These systems are spin-1/2 lattice models, in which local fermionic
operators are introduced by the Jordan-Wigner transformation  \cite{jordan-wigner}. The relevant Hamiltonians can be
expressed as bi-linear in these fermionic operators, and the diagonalization can be performed effectively by finding the
eigenmodes in the fermionic space. In many-body physics a number of works considered generalizations of the JW
transformation, see for example \cite{japan-JW-gen-1,japan-JW-gen-2,japan-JW-gen-3,japan-free-fermion-JW,chapman-jw}.

In the recent work \cite{ffd} Fendley introduced a new spin chain model, dubbed ``free fermions in disguise'' (FFD). This
model was shown to be free fermionic
\cite{ffd}, but its solution lies beyond the Jordan-Wigner paradigm \cite{japan-free-fermion-JW,chapman-jw,fermions-behind-the-disguise}. The model is the first and simplest one in a large
family of new lattice models, all of which have a hidden free fermionic structure
\cite{alcaraz-medium-fermion-1,alcaraz-medium-fermion-2,fermions-behind-the-disguise,free-fermion-subsystem-codes,unified-graph-th,sajat-FP-model,sajat-claws}
(for earlier examples see \cite{cooper-anyon,gyuri-susy-1,gyuri-susy-2}). The models typically have a large non-abelian
symmetry algebra; in the case of the FFD model the algebra was worked out recently in \cite{eric-lorenzo-ffd-1}.

In these models the fermionic operators are highly non-local in both the spin basis and also when expressed using
the standard Jordan-Wigner fermions. Therefore it was not obvious that they could be useful for practical calculations
beyond diagonalizing the Hamiltonian. However, it was shown in \cite{sajat-ffd-corr} that a subset of local operators can be
expressed using the fermionic eigenmodes, leading to exact results for correlation functions. These results can also be used in
special quantum circuits \cite{sajat-floquet} which retain free fermionic structures \cite{sajat-claws,eric-lorenzo-circuits}.

Remarkably, the formalism of \cite{ffd,fermions-behind-the-disguise,unified-graph-th} (see also \cite{free-parafermion})
appears to be disconnected from
traditional methods for free fermions. The new method was only applied to systems with open boundary conditions, and
even in these cases the connection to standard Jordan-Wigner solvability has not yet been made explicit.

Furthermore, despite all the progress in this topic, explicit eigenstates in such models have not yet been constructed. The
reason for this lies in the absence of a reference state for the original Hamiltonian of Fendley. A similar situation
also persists for many Jordan-Wigner solvable models: if there is no reference state, then there is no notion of a
few-body wave function. Instead, eigenstates could be constructed by taking a random vector in the Hilbert state,
acting with a fermionic creation or annihilation operator for every eigenmode. For example, the ground state can be
reached by acting with all annihilation operators \cite{takahashi-book}.

\bigskip

We address these questions in a series of papers, consisting of three parts.

This is the first paper in the series. Here we treat the abstract algebras behind these models. They are called
graph-Clifford algebras  \cite{graph-clifford} or quasi-Clifford algebras \cite{quasi-clifford-phd,quasi-clifford} (see
also \cite{huber-quasi-clifford}). Following  \cite{graph-clifford} we establish a number of key theorems about the
structure of these algebras, their operator content, and their representations. In particular, we introduce the
{\it defining representation} for the tensor product of two copies of a graph-Clifford algebra, and show that the XY
model and the FFD model have Hamiltonians which originate from these representations. Afterwards we study a particular
anti-symmetric combination of two commuting Hamiltonians. We show that it has a reference state, and few body
excitations can be created on top this reference state by acting with  the fermionic operators. 

In the second paper we develop a new representation of the fermionic operators, which we call the {\it path-product
  expansion}.  We also present new and independent
proofs for the fermionic solvability of the models in question. Furthermore, we study the spatial structure of the
operators and eigenstates in the FFD model.

In the third paper we establish a connection with the well known methods applied in Jordan-Wigner solvable models. We
show that using the results of the second paper, the standard formulas for the fermionic operators are reproduced in
selected cases. However, we also encounter a surprising result: in a particular version of the anisotropic (and possibly
inhomogeneous) XY model our formalism actually leads to the construction of fermionic states in the so-called chiral
basis \cite{xy-helix-basis}. This basis is obtained by adding excitations over spin helices.

\section{Setting and earlier results}

\label{sec:prev}

In this Section we introduce the framework and summarize previously known results.

\subsection{Models and frustration graphs}

We are dealing with lattice models, where a Hamiltonian is defined as
\begin{equation}
  \label{Hdef}
  H=\sum_{j=1}^M  b_jh_j.
\end{equation}
Here $b_j\in\valos$ are coupling constants, and the $h_j$ are Hermitian operators, which are assumed to be localized in
a lattice. We often call the $h_j$ as {\it
  generators} of particular algebras. They satisfy special relations, which are motivated by experience from spin-1/2
systems.

In a spin-1/2 lattice model the space of operators is spanned by the Pauli-strings: product of Pauli operators acting on
different sites. Every Pauli string is Hermitian, it squares to the identity, and two Pauli strings either commute or
anti-commute. This gives the motivation for the following abstract relations, which we require independent of
representation:
\begin{equation}
  \begin{split}
    h_j^\dagger=h_j,\quad (h_j)^2=1,\qquad j=1\dots M
  \end{split}
\end{equation}
and
\begin{equation}
  h_jh_k=(-1)^{A_{j,k}}h_kh_j.
\end{equation}
Here $A_{j,k}=0, 1$  encodes whether we have commutation or anti-commutation. Hermiticity is required only for the
physical interpretation as a Hamiltonian, and most of our computations do not require this property.

The structure of such  an algebra can be conveniently encoded in the {\it frustration graph}, which will be denoted
generally as $G$. It is defined as the
graph with $M$  vertices, where two vertices $j$ and $k$ are connected if the corresponding generators anti-commute,
i.e. $A_{j,k}=1$. In such a cases we have
``frustration'' or ``interaction'' between the two terms in the Hamiltonian, because they can not be diagonalized
simultaneously.

Such algebras are known in the mathematical literature as {\it quasi Clifford algebras}
\cite{quasi-clifford-phd,quasi-clifford} or {\it graph-Clifford algebras} \cite{graph-clifford}. In the following Section
 we discuss the mathematical structure of these algebras.
Here we just put forward that the abstract algebra has finite dimension $2^M$, and it is
always possible to find concrete matrix algebra representations.
We will work with representations on Hilbert spaces of spin chains, and $L$ will denote the length of the spin chain. 
In Section \ref{sec:graph-clifford} we introduce the ``defining representation'' on a spin chain with length $L=M$. However,
often it is possible to find representations with smaller $L$. 

The key properties of the model defined by \eqref{Hdef} are independent of the presentation of the algebra, and will
depend only on
the frustration graph. This includes the integrability of the model, the free fermionic nature, and the concrete
numerical eigenvalues of the Hamiltonian. However, the degeneracies of the eigenvalues of the Hamiltonian can depend on
the presentation of the algebra.

Interestingly, the key properties of integrability and also free fermionic nature will not depend on the coupling
constants $b_j$; they can be regarded as arbitrary numbers. This is in stark contrast with usual integrable
models, which require a spatial homogeneity of the coupling constants.

In many of our examples we will also have a second copy of the same algebra, with generators $\tilde h_j$. These generators
satisfy exactly the same relations. In other words, the models defined by $h_j$ and $\tilde h_j$ have the same
frustration graph. Furthermore, the two copies will be independent from each other:
\begin{equation}
  \label{htildeh}
  [h_j,\tilde h_k]=0.
\end{equation}
We will also investigate the Hamiltonian
\begin{equation}
  \label{Htdef}
 \tilde H=\sum_{j=1}^M  b_j\tilde h_j
\end{equation}
and the anti-symmetric combination
\begin{equation}
  \label{HAdef}
  H_A=H-\tilde H=\sum_{j=1}^M   b_j(h_j-\tilde h_j).
\end{equation}
We will see that we can compute exact eigenstates in a simple form for the anti-symmetric Hamiltonians.

\subsection{Examples}

Here we mention three specific examples for the abstract framework. These are the XY model (in a particular presentation),
the quantum Ising chain, and original ``free fermions in disguise'' model of Fendley.

\subsubsection{The XY model via the Dzyaloshinskii–Moriya interaction}

Let us consider a Hamiltonian given solely by the Dzyaloshinskii–Moriya interaction term:
\begin{equation}
  \label{XYH}
  H_A=\sum_{j=1}^M (Y_jX_{j+1}-X_jY_{j+1}).
\end{equation}
We impose periodic boundary conditions, and for simplicity we used here homogeneous couplings.

We identify the generators $h_j$ and $\tilde h_j$ as
\begin{equation}
  h_j=Y_{j-1} X_{j},\quad \tilde h_j=X_j Y_{j+1}.
\end{equation}

It is easy to see that \eqref{htildeh} holds, and that the two sets of generators satisfy the same algebra.
The algebra is given by
\begin{equation}
  \label{isingXY}
  \begin{split}
    \{h_j,h_{j+1}\}&=0,\\
    [h_j,h_k]&=0,\quad\text{if}\quad |j-k|>1,
  \end{split}
 \end{equation}
and in this case periodic boundary conditions are understood for the indices.
 The frustration graph is a circle of length $M$, and it is shown in
Figure \ref{fig:circle}.

\begin{figure}[t]
  \centering
  \begin{tikzpicture}
  % ellipse radii (half-width and half-height)
  \def\a{2}    % semi-major axis (x)
  \def\b{1}  % semi-minor axis (y)
  \def\n{7}    % number of vertices

  % create coordinates for the n vertices placed evenly around an ellipse
  \foreach \i in {0,...,6} {
    \pgfmathsetmacro{\angle}{360/\n * \i} % degrees
    \coordinate (v\i) at ({\a*cos(\angle)},{\b*sin(\angle)});
  }

\draw[thick] (0,0) ellipse ({\a} and {\b});

\foreach \i in {0,...,6} {
    \node[draw, thick,circle, fill=white, inner sep=2pt] at (v\i) {};
  }

\end{tikzpicture}
  \caption{Frustration graph for the XY and Ising models with periodic boundary conditions.}
  \label{fig:circle}
\end{figure}
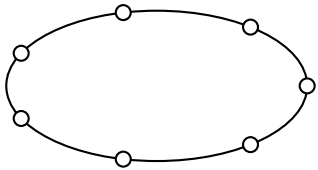

\subsubsection{The quantum Ising chain}

We consider $M=2L$ and define the Hamiltonian acting on a spin-1/2 chain of length $L$ as
\begin{equation}
  H=\sum_{j=1}^L  (X_jX_{j+1}+g Z_{j+1}).
\end{equation}
We assume periodic boundary conditions, and also homogeneous couplings. The generators can be identified as
\begin{equation}
  \label{IsingH}
  h_{2j-1}=X_jX_{j+1},\quad h_{2j}=Z_{j+1}.
\end{equation}
Inspecting the commutation relations of the generators we see that they satisfy the same algebra as in the previous
case: Neighboring generators anti-commute, otherwise they commute.  Therefore, the frustration graph is the circle of
length $M=2L$.

A crucial difference is that now we can not introduce an independent set of generators $\tilde h_j$ within the same
representation. This has to do with
the fact that now the representation acts on the Hilbert space of a chain with length $L=M/2$, and
the $h_j$ generate almost the complete operator algebra of the chain (a precise statement will be given
later in Subsection \ref{sec:isingxy}).

\subsubsection{Free fermions with all-to-all interactions}

Let us now consider a family of Majorana fermions $\chi_a$, with $a=1\dots M'$. They are defined via the relations
\begin{equation}
  \chi_a^\dagger=\chi_a,\qquad \{\chi_a,\chi_b\}=2\delta_{a,b}.
\end{equation}
If we would construct the frustration graph for this family of operators, we would obtain a complete graph with $M'$
vertices. However, instead of taking linear combinations of the $\chi_a$, we consider Hamiltonians that are bilinear in the
Majorana fermions. In this case it is convenient to index the terms of the Hamiltonian with a pair $(a,b)$, such that
$a<b$, and we write
\begin{equation}
  H=\sum_{a<b} B_{ab} h_{ab},\qquad h_{ab}=i\chi_a\chi_b.
\end{equation}
We have $M=M'(M'-1)/2$ terms, the $B_{ab}$ are real coupling constants. Note that the generators of the Hamiltonian are
Hermitian, due to the factor of $i$ above.

The frustration graph of this model is the so-called {\it line graph} associated to the complete graph. Now the vertices
are given by the pairs $(a,b)$, which are the edges  of the original complete graph.
In the line graph two different pairs $(a,b)$ and $(c,d)$ are connected by an edge, if
they share exactly one original vertex. If the four indices $a,b,c,d$ are all different, then the two terms $h_{ab}$ and
$h_{cd}$ commute with each other, thus the pairs are not connected in the line graph.

$M'=2L$ Majorana fermions can be realized on a spin chain of length $L$ by the well known Jordan-Wigner transformation:
\begin{equation}
  \begin{split}
    \chi_{2j-1}&=Z_1Z_2\dots Z_{j-1}Y_{j}\\
    \chi_{2j}&=Z_1Z_2\dots Z_{j-1}X_{j}.\\
  \end{split}
\end{equation}
The product of $Z$ operators that stretch from site 1 to site $j-1$ is called the {\it Jordan-Wigner string}.

\subsubsection{Free fermions in disguise}

The model in question was introduced by Paul Fendley in the seminal work \cite{ffd}.

Let us consider open boundary conditions, and the spin chain of length $L=M+2$. We introduce generators
\begin{equation}
  \label{bulkrep1}
  h_j=Z_{j-2}Z_{j-1}X_j,\qquad \tilde h_j=X_j Z_{j+1} Z_{j+2}.
\end{equation}
The two sets of generators commute with each other, and they satisfy the same
algebra. This algebra is given by the commutation relations
\begin{equation}
  \label{ffdalg}
  \begin{split}
    \{h_j,h_{j+1}\}=\{h_j,h_{j+2}\}=0\\
    [h_j,h_k]=0,\quad |j-k|>2\,.
  \end{split}
\end{equation}
The frustration graph of this algebra is depicted on Fig. \ref{fig:ffd} for $M=6$.

It is also possible to consider the model with periodic boundary conditions. That model is integrable, but it has not
yet been solved \cite{ffd}. In that case the frustration graph is a so-called circulant graph.

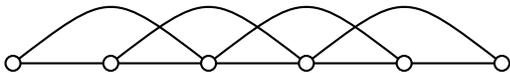
\begin{figure}[t]
  \centering

\begin{tikzpicture}
  \def\n{6}         % number of vertices
  \def\dx{1.3}      % spacing between vertices
  \def\arcB{1}    % arc height for next-to-neighbours

  % coordinates of vertices along a line
  \foreach \i in {1,...,\n} {
    \coordinate (v\i) at ({(\i-1)*\dx},0);
  }

  % draw edges: neighbours (distance 1) as horizontal lines
  \foreach \i in {1,...,5} {
    \pgfmathtruncatemacro{\j}{\i+1}
    \draw[thick] (v\i) -- (v\j);
  }

  % draw edges: next-to-neighbours (distance 2) as arcs above
  \foreach \i in {1,...,4} {
    \pgfmathtruncatemacro{\j}{\i+2}
    \path let \p1=(v\i), \p2=(v\j) in
      coordinate (c) at ({0.5*(\x1+\x2)},{\arcB});
    \draw[thick] (v\i) .. controls (c) .. (v\j);
  }

  % draw the vertices as small circles
  \foreach \i in {1,...,\n} {
    \node[draw,circle,fill=white,inner sep=2pt,thick] at (v\i) {};
  }

\end{tikzpicture}

  \caption{Frustration graph for the FFD model with $M=6$.}
  \label{fig:ffd}
\end{figure}

\subsection{Integrability}

Now we
discuss the integrability of the
Hamiltonian $H$ defined by \eqref{Hdef}.

Perhaps surprisingly, there is a single sufficient condition on the frustration graph, which ensures the existence of a
large family of conserved charges, independent of the coupling constants. This condition is that the graph should be claw-free
\cite{ffd,fermions-behind-the-disguise,unified-graph-th}. A claw is a special induced sub-graph, where 3 vertices are
connected to a central vertex, so that no two of the other three are connected (see Fig. \ref{fig:claw}). In terms of the Hamiltonian
the condition simply means that if there is a $h_a$ which anti-commutes with three different other generators, then it
must not happen that all three are commuting with each other. It was shown in
\cite{ffd,fermions-behind-the-disguise,unified-graph-th} that if the graph is claw-free, then the models are integrable
{\it for any choice of the coupling constants}.

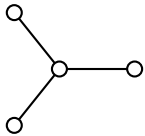
\begin{figure}[t]
  \centering
  \begin{tikzpicture}
  % central vertex
  \coordinate (c) at (0,0);

  % three outer vertices placed around
  \coordinate (v1) at (1,0);
  \coordinate (v2) at (-0.6,-0.75);
  \coordinate (v3) at (-0.6,0.75);

  % draw edges from center to each outer vertex
  \draw[thick] (c) -- (v1);
  \draw[thick] (c) -- (v2);
  \draw[thick] (c) -- (v3);

  % draw nodes
  \node[thick,draw,circle,fill=white,inner sep=2pt] at (c) {};
  \node[thick,draw,circle,fill=white,inner sep=2pt] at (v1) {};
  \node[thick,draw,circle,fill=white,inner sep=2pt] at (v2) {};
  \node[thick,draw,circle,fill=white,inner sep=2pt] at (v3) {};
\end{tikzpicture}
  \caption{A ``claw'' graph.}
  \label{fig:claw}
\end{figure}

This situation is to be contrasted with that of the most common integrable spin chains. For example, the Hamiltonian of
the standard Heisenberg XXZ model is written as
\begin{equation}
  H=\sum_{j=1}^L (X_jX_{j+1}+Y_{j}Y_{j+1}+\Delta Z_j Z_{j+1}).
\end{equation}
Here we can identify the Hamiltonian generators as
\begin{equation}
  h_{3j-2}=X_jX_{j+1},\quad  h_{3j-1}=Y_jY_{j+1},\quad  h_{3j}=Z_jZ_{j+1}.
\end{equation}
The resulting frustration graph has claws. For example, the operator $Z_{2}Z_3$ anti-commutes with $X_1X_2$, $Y_1Y_2$
and $X_3X_4$, while all three operators are commuting with each other. Therefore, the XXZ model does not fit into our
framework. This is in accordance with a key requirement of integrability, namely that the coupling constants
should be homogeneous. Instead, in our framework all couplings are free parameters, which is typical for free fermionic
models.

\subsubsection{Higher charges and transfer matrix}

Let us now construct the higher conserved charges of the models that we treat. For future use we introduce the notation
\begin{equation}
  {\mathtt H}_k=b_k h_k.
\end{equation}
We say that $Q_1$ is the Hamiltonian, in these notations
\begin{equation}
  Q_1=\sum_k   {\mathtt H}_k.
\end{equation}
The next conserved charge is given by
\begin{equation}
  Q_2=\mathop{\sum_{j<k}}_{A_{jk}=0}  {\mathtt H}_j {\mathtt H}_k.
\end{equation}
Note that the sum runs over those products where the two Hamiltonian terms commute with each other. This pattern can be
extended, and we define the higher charges as
\begin{equation}
  \label{Qalpha}
  Q_\alpha=\mathop{\sum_{a_1<\dots<a_\alpha}}_{A_{a_ja_k}=0}  {\mathtt H}_{a_1} {\mathtt H}_{a_2}\dots  {\mathtt H}_{a_\alpha}.
\end{equation}
where it is implicitly understood that we are summing over those products, where every term commutes with every other
term.  In terms of graph theory concepts, we are performing a sum over the so-called {\it independent sets} of the graph
$G$.

Such charges exist up to order $\alpha=S$, where $S$ is the so-called {\it independence number} of the graph. This is
the maximum number of vertices one can choose from the graph such that no two chosen vertices are connected by an edge.

Finally the transfer matrix is defined as the generating function
\begin{equation}
  \label{Tudef}
  T(u)=\sum_{n=0}^S (-u)^n Q_n.
\end{equation}
It was shown in \cite{ffd,fermions-behind-the-disguise,unified-graph-th} that if $G$ is claw free, then all these
charges commute with each other, therefore
\begin{equation}
  [T(u),T(v)]=0.
\end{equation}

We stress that claw-freeness is a sufficient, but not a necessary condition. There can be Hamiltonians where
similar structures can be developed, despite the graph having claws \cite{sajat-FP-model,sajat-claws}. In those cases
there are special relations between
the coupling constants, which ensure the commutativity of the charges.

The reader might wonder what is the connection between this formalism for the charges  and the transfer matrix, and
other, more traditional formulas in the theory of
integrable models. We will argue, that despite the apparent differences, in many models we can establish connections to
earlier results.

Let us consider for example the transfer matrices in standard Algebraic Bethe Ansatz \cite{faddeev-how-aba-works}. These
transfer matrices are usually given by a Matrix Product Operator (MPO) that uses the so-called Lax operator as the
fundamental tensor. It turns out that the $T(u)$ of \eqref{Tudef} can often be described also as an
MPO's (in the case of the FFD model see \cite{ffd} and Appendix \ref{sec:MPOT}).
A crucial point is that in standard integrability one typically considers periodic boundary conditions and
homogeneous couplings, and in those cases one obtains a larger set of conserved charges.

We should also note, that conserved charges with localized densities can be obtained by taking the logarithmic
derivatives of $T(u)$  \cite{ffd}. Currently there are no explicit expressions available in the
literature for these local charges.

\subsubsection{Inversion relation}\label{subsec:invrel}

Further useful properties can be derived if the graph $G$ satisfies another special relation, namely that it is
even-hole-free. An even hole (also called chordless cycle) is a cycle in the graph of even length, where only
consecutive vertices are connected, but there are no ``chords''.

It was shown in \cite{fermions-behind-the-disguise} that if $G$ is claw-free and even-hole-free, then the transfer
matrix satisfies an inversion relation
\begin{equation}
  \label{eq:inversion}
  T(u)T(-u)=P_G(u^2)\cdot \mathbb{1}\,.
\end{equation}
Here $P_G(x)$ is a polynomial of order $S$, and the index $G$ stands for the graph. The polynomial $P_G(x)$ is the weighted
independence polynomial of the graph. It can be evaluated by the essentially the same formulas as the charges and
the transfer matrix, with the only difference that the operators $H_k$ are to be replaced by the numbers
$H_k^2=b_k^2$. Concrete examples will be considered below.

In those cases, when the graph is claw-free but it contains even holes, the inversion relation becomes operator-valued,
where so-called  {\it cycle symmetries} appear on the r.h.s. of \eqref{eq:inversion}, see \cite{unified-graph-th}. In
this paper we do not consider these cases in
detail.

In the next Subsection we will show that the roots of these polynomials play a crucial role in solving the model: they
are directly related to the single-particle energy eigenvalues of the free fermionic models.

{In the models that appeared in the literature} 
the transfer matrices and also the
polynomials $P_G(x)$ can be obtained in a very effective way, using simple recursion relations. This means that they can
be obtained in polynomial time in the size of the graph $G$.

\subsection{Free fermions behind the disguise}

\label{sec:fermion-intro}

If  the frustration graph $G$ is claw-free and even-hole-free, then the Hamiltonian can be diagonalized using free
fermionic creation and annihilation operators \cite{ffd,fermions-behind-the-disguise}. In this Section we summarize this
construction.

The resulting formulas have a very different structure than the usual expressions for creation/annihilation
operators in traditional Jordan-Wigner solvable models \cite{XX-original,Schulz-Mattis-Lieb}. We reconcile the apparent
differences later follow up papers in this series.

\subsubsection{\label{sec:simplicial}Simplicial cliques}

One of the key ingredients of the free fermionic solution is a special subgraph of the frustration graph: the simplicial clique.
A simplicial clique is a clique (a subgraph which is a complete graph),
such that for every vertex within the clique, the neighbors of the given vertex outside the clique also form a
clique. It was proven in the work \cite{evenhole-simplicial} that a claw-free and even-hole-free graph always
contains a simplicial clique.

In practical terms, the simplicial clique is very often just one vertex in the graph, corresponding to a special
generator $h_k$. In this case the
requirement of the simplicial clique is that all other generators that anti-commute with $h_k$, should also mutually
anti-commute with each
other.

It was shown in the work \cite{unified-graph-th} that a generalized free fermionic solution exists even in those cases
when the graph $G$ contains even holes, if it also contains a simplicial clique. In these cases the Hilbert space splits
into different sectors, corresponding to different eigenvalues of so-called cycle symmetries. The actual free fermionic
solution is then constructed for each sector separately.
The
Jordan-Wigner solvable models with periodic boundary conditions belong to this class. In those models there are two
sectors, corresponding to odd and even fermion numbers.

\subsubsection{Construction of the fermionic operators}

Let us now turn to the diagonalization of the Hamiltonian $H$, assuming that the graph $G$
is claw-free and also even-hole-free, and also that a simplicial clique has been identified.

The construction will use a certain edge operator $\chi$. This is a
Hermitian operator satisfying $\chi^2=1$. It has to anti-commute with the generators in the simplicial clique, but it
should commute with all other generators. In models with open boundary conditions this operator is often localized to one of
the boundaries, hence its name. However, it is not necessarily  localized  in the 
physical space, it is only required that it is sufficiently local with respect to commutation with the generators.

In some cases $\chi$ can be chosen to lie within the algebra generated by the $h_k$. In other cases one can either embed
$\chi$ naturally into the operator algebra in the concrete representation,
or one can construct an extended representation with an additional spin \cite{unified-graph-th}. We will also encounter such
examples below.

If all ingredients are given, then the model can be solved by Dirac fermions as follows.

It was proven in
\cite{ffd,fermions-behind-the-disguise,unified-graph-th}
that  there exist Dirac fermions $\Psi_k$ and $\Psi_{-k}\equiv \Psi_k^\dagger$, $k=1, 2, \dots S$,
such that
\begin{equation}
  \label{Diracdef}
  \{\Psi_k,\Psi_\ell\}=\{\Psi_{-k},\Psi_{-\ell}\}=0,\quad  \{\Psi_k,\Psi_{-\ell}\}=\delta_{k,\ell},
\end{equation}
which diagonalize the Hamiltonian via
\begin{equation}
  \label{Hexpress}
  H=\sum_{k=1}^S   \eps_k [\Psi_k,\Psi_{-k}]\,.
\end{equation}
Here the $\eps_k$ are interpreted as (half of the) single-particle energies.
These energies can be computed
 by $\eps_k=(u_k)^{-1}$, where $u_k$ are the positive roots of a polynomial $P_{G}(u^2)$.

Using the edge operator the explicit formula for the fermionic operators is \cite{ffd}
\begin{equation}
  \label{Psidef}
  \Psi_{\pm k}=\frac{1}{\NN_k}T(\mp u_k) \chi T(\pm u_k),
\end{equation}
where $\NN_k$ is a known normalization factor given by
 \begin{align}\label{eq:normalization-factor-Psi}
    \mathcal{N}_k
     & =
       4 \sqrt{- u_k^2 P_{G \setminus K}(u_k^2) P'_{G}(u_k^2)}.
  \end{align}
Here $P'_G(x)$ denotes the derivative of the polynomial with respect to $x$. By $K$ we denote the simplicial clique
discussed in Section \ref{sec:simplicial}. The polynomial $P_{G\setminus K}(x)$ then corresponds to a frustration graph
$G$ with the simplicial clique removed.

A consequence of this diagonalization is that the eigenvalues of $H$ can be written as
\begin{equation}
  \label{energyformula}
  E=\sum_j \sigma_j \eps_j,
\end{equation}
where $\sigma_j=\pm$ and the signs can be chosen independently.

It is important that there are only $S$ fermionic eigenmodes for a given size $M$ of the algebra. Typically $S$ is much
smaller than the linear size of the spin chain on which the model is defined. In fact, very often $S$ grows linearly
with $M$, but the ratio is such that $S/L<1$, where $L$ is the length of the spin chain on which
the representation 
lives.
This implies that
each energy level is degenerate with the same level of degeneracy which increases exponentially with
$M$. The actual size of the degenerate sectors is $2^{L-S}$, which depends on the choice of representation, determining
the volume
$L$ for a given $M$. Exceptions to exponential degeneracies are given by the models solvable by Jordan-Wigner
transformation.

If there is an exponential degeneracy, then it implies an exponentially growing symmetry algebra behind the model. In
the case of the original ``free fermions in disguise'' model of Fendley this algebra was worked out recently in
\cite{eric-lorenzo-ffd-1}.

The previous results only concern the spectrum, but not the eigenvectors. Up to now there were practically no results in
the literature for the eigenvectors of such disguised free fermionic models. The reason for this is simple: typically
there is no reference state, which could serve as a starting point for the creation of fermionic wave functions.
This happens for example in the case of the quantum Ising chain.

One of our key results is that very often the anti-symmetric combination $H_A$ \eqref{HAdef}
does in fact have simple reference states. We thus propose to study this particular combination, and eigenstates are
constructed later in Section \ref{sec:eigenstates}.

\section{Graph-Clifford algebras and the defining representation}

\label{sec:graph-clifford}

Here we review basic statements about the structure of graph-Clifford algebras, following
\cite{quasi-clifford-phd,quasi-clifford,graph-clifford}. We also introduce a representation for the tensor product of
two identical algebras; this appears to be a new result of this work. We call it the ``defining representation'',
because it is obtained by letting the algebra act on itself.

\subsection{The structure of graph-Clifford algebras}

We say that an associative algebra $\AAA$ is a graph-Clifford algebra, if it is generated by $M$ elements $h_k$, $k=1,\dots,M$,
satisfying
\begin{equation}
  h_k^2=1,\qquad h_kh_j=(-1)^{A_{jk}} h_jh_k,
\end{equation}
where $A_{jk}=A_{kj}=0, 1$. The requirement of Hermiticity is not included in this definition; previously we
used it for the physical interpretation of the Hamiltonians.

We construct a graph with $M$ vertices, by adding an edge between vertices $j$ and $k$
precisely if $A_{jk}=1$. The symmetric matrix $A$ given by the elements $A_{jk}$ is the {\it adjacency matrix} of the graph.

It is an important consequence of the defining relations that the graph-Clifford algebras are finite dimensional as an abstract
algebra. Furthermore, every element in the algebra can be written as a linear combination of simple ordered
products. For the moment let
us focus on only one copy of the algebras. Then a basis of the algebra is given by the operator products
\begin{equation}
  \label{basiselement}    %previously: normalordered
  h_{a_1}h_{a_2}\dots h_{a_n},
\end{equation}
where $a_1<a_2<\dots<a_n$. This can be shown by taking an arbitrary finite product of the generators, commuting them to
obtain the desired ordering, and noting that in the final expression every $h_a$ can be present at most only once, due
to the relation $(h_a)^2=1$. It follows that the abstract algebra has dimension
$2^M$. Furthermore, its basis elements can be indexed by a bit string of length $M$, where we add a bit 1 or 0
corresponding to whether a specific generator $h_a$ is present in the product or not. We will use this correspondence to
bit strings in Section \ref{sec:opstate} where we develop an operator-state correspondence.

Every two basis elements either commute or anti-commute with each other. In order to determine the outcome of the
commutation one needs to collect the minus signs that arise. This can be read off from the adjacency matrix.

We introduce the concept called {\it commutation pattern}.
The commutation pattern of a basis element is a vector of length $M$, consisting of 0's and 1's, such that the $k$th
element describes the commutation/anti-commutation with the generator $h_k$. The commutation pattern of a generator $h_j$
coincides with the $j$th row of the matrix $A$. The commutation pattern of the basis element
\eqref{basiselement} is the sum of $n$ rows of $A$, corresponding to the indices $a_1, a_2, \dots, a_n$, evaluated
in $\Zk$.

Let us now present a few key theorems about the structure of graph-Clifford algebras, following
\cite{graph-clifford}. These theorems will be used in 
later Sections, where we discuss the faithfulness (and completeness) of certain representations, and the existence of
the edge operators. 

\begin{thm}
\cite{graph-clifford}  If the determinant of $A$ is an odd number (in $\egesz$), then the center of the algebra is trivial.
\end{thm}
\begin{proof}
We prove the theorem by contradiction. Let us assume that there exists a non-trivial central element in the
algebra. Then there must be a central element of the 
form \eqref{basiselement}. A central basis element has a commutation pattern given by zeros only. This element is
a product of the generators $h_k$,
therefore its commutation pattern is a sum over some rows of $A$ over $\Zk$. This implies that the determinant of $A$
over $\Zk$ is zero, therefore it is an even number over $\egesz$.
\end{proof}

\begin{thm}
  \label{thm:non-degenerate}
  If the determinant of $A$ is an odd number, then the algebra is non-degenerate, in the sense that every possible
  commutation pattern is present in $\AAA$.
\end{thm}
\begin{proof}
  The algebra has dimension $2^M$, which equals the number of all possible commutation patterns. If there is a
  commutation pattern which can not be found among the basis elements, then it means that two
different basis  elements have the same commutation pattern. Taking the product of these elements, we obtain an element
which has a commutation pattern consisting of 0's, therefore it is in the center. If $\det(A)$ is odd, then the center
is trivial, and with this we proved the theorem by contradiction.
\end{proof}

This implies that in the non-degenerate cases we can always embed any kind of edge operator into the algebra.

An even stronger theorem holds about the structure of graph-Clifford algebras. We don't use the general theorem here,
but we cite it for completeness:
\begin{thm}
  \cite{graph-clifford}
  Two graph-Clifford algebras are isomorphic, if the number of generators and the dimension of the center coincide.
\end{thm}
\begin{proof}
  The proof is presented in \cite{graph-clifford}.
\end{proof}

\subsection{The defining representation}

\label{sec:opstate}

Here we establish an {\it operator-state correspondence}, and we also construct the {\it defining representation}
of the tensor product $\AAA\otimes\AAA$ of two copies of the algebra.
This representation is motivated by the examples of the XY and FFD models. However, in this
general form it appears to be new.

We proceed in three steps. First, we introduce the {\it operator--state correspondence} by mapping the
abstract algebra $\AAA$ to the Hilbert space of an auxiliary spin chain. Next, we derive how left- and
right-multiplication by the generators are represented as Pauli strings acting on this auxiliary chain. Finally, we
observe that the two actions commute, and therefore they provide a natural representation of the tensor product
$\AAA\otimes\AAA$.

\paragraph*{Operator--state correspondence.}
We define a linear map
$
  \varphi:\AAA\to \bigotimes_{j=1}^M \complex^2
$
by specifying its action on the standard basis elements \eqref{basiselement}. We identify the computational basis of the
auxiliary chain with bit strings, and we adopt the convention that $\ket{0}$ (resp.\ $\ket{1}$) corresponds to an up (resp.\ down) spin. For an ordered basis element $h_{a_1}h_{a_2}\dots h_{a_n}$ with $a_1<a_2<\dots<a_n$, we set
\begin{equation}
  \label{eq:phi-def}
  \varphi\!\left(h_{a_1}h_{a_2}\cdots h_{a_n}\right)
  =
  \ket{s_1 s_2 \cdots s_M}
  ,
\end{equation}
where $s_i = 1$ if $i \in \{a_1, a_2, \ldots, a_n\}$ and $s_i = 0$ otherwise. Equivalently, we can write $ \varphi\!\left(h_{a_1}\cdots h_{a_n}\right)
  =
X_{a_1}\cdots X_{a_n}
\ket{\emptyset}$
where $\ket{\emptyset} \equiv \ket{00\cdots0}$ is the all-up state and $X_{a}$ is the Pauli $X$ operator acting on the
$a$-th site of the auxiliary chain.
With this convention,
every ordered basis element is mapped to a computational basis state with coefficient $+1$, while other orderings can
differ by a sign. 
For a generic element of $\AAA$ we first expand it in the basis \eqref{basiselement} (using the anticommutation
relations to bring products into the standard order), and then extend $\varphi$ by linearity. 

Here we present a few simple examples of the operator-state correspondence for the case $M=4$.
The identity element is mapped to the all-up state:
\begin{equation}
\varphi(1)=\ket{0000}.
\end{equation}
Moreover, the basis element $h_2h_4$ corresponds to down spins at sites $2$ and $4$:
\begin{equation}
\varphi(h_2h_4)=\ket{0101}.
\end{equation}
For the non-ordered product $h_4h_2$, then using the relation $h_4h_2=(-1)^{A_{24}}h_2h_4$ we obtain
\begin{equation}
\varphi(h_4h_2)=(-1)^{A_{24}}\ket{0101},
\end{equation}
which illustrates how the commutation relations are translated into sign factors in the auxiliary Hilbert space.

Let us now consider linear operations on basis elements in $\AAA$. More concretely, we consider the left- and
right-multiplication by any generator $h_k$, and we derive how to represent these linear operations on the states of the
auxiliary spin chain. Once these actions are known, they can be extended to the most general linear operations within
$\AAA$.

\begin{thm}
  \label{thm:def-rep}
Left-multiplication with $h_k$ is represented by
\begin{equation}
  \label{leftm}
X_k\mathop{\prod_{j<k}}_{A_{jk}=1}Z_j
\,.
\end{equation}
Right-multiplication with $h_k$ is represented by
\begin{equation}
  \label{rightm}
X_k\mathop{\prod_{j>k}}_{A_{jk}=1}Z_j
\,.
\end{equation}
\end{thm}
\begin{proof}
  Consider the operation of multiplying an operator product in the algebra by the left:
  \begin{equation}
    h_k (h_{a_1}h_{a_2}\dots h_{a_n}).
  \end{equation}
This operator product can be transformed into one of the basis states \eqref{basiselement}. In order to find out which,
we need to perform a normal ordering, and move $h_k$ to position $k$. In performing the normal ordering we can pick up
signs, which depend on whether there are operators in the product, which could anti-commute with $h_k$. Possible signs
that arise are generated by the $Z$ operators acting on the auxiliary chain. Once $h_k$ is moved to its position there are
two possibilities: either $h_k$ was already present in the product or not yet. In the first case we get the identity
due to $(h_{k})^2=1$, in the second case we get simply $h_k$. The combination of both cases can be represented by the
operator $X_k$ acting on the auxiliary chain. This concludes the proof.

For the case of right-multiplication one needs to perform an analogous argument, in which case we can pick up phases from
the ``right'' of site $k$.
\end{proof}

Both \eqref{leftm} and \eqref{rightm} can be used as a representation of $\AAA$. However, since the left- and
right-multiplications commute with each other, it is natural to regard the formulas as a combined representation of the
tensor product $\AAA\otimes\AAA$. More precisely, elements of the form $h_k\otimes 1$ are represented by \eqref{leftm}
and elements of the form $\tilde h_k\equiv 1\otimes h_k$ by \eqref{rightm}.
With this we established the ``defining representation''.

The formulas \eqref{leftm}-\eqref{rightm} have some similarities with the Jordan-Wigner transformation. In fact, in the
special case when the
generators $h_k$ are Majorana fermions (every two generators anti-commute), and the graph $G$ is the complete graph over
$M$ vertices, then the formulas above are identical to the standard Jordan-Wigner transformation.

As further examples let us consider the XY/Ising and FFD algebras. In the case of the XY/Ising algebra (as given by
eq. \eqref{isingXY} without periodic boundary conditions) we obtain
\begin{equation}
  \begin{split}
    h_1&=X_1,\quad h_j=Z_{j-1}X_j,\quad j>1   \\
\tilde    h_M&=X_M,\quad \tilde h_j=X_j Z_{j+1},\quad j<M.
  \end{split}
\end{equation}
We could perform a global rotation to arrive at the representation
\begin{equation}\label{eq:XYrot}
  \begin{split}
	  h_1&=Y_1,\quad h_j={X}_{j-1}Y_j,\quad j>1   \\
	  \tilde    h_M&=Y_M,\quad \tilde h_j=Y_j {X}_{j+1},\quad j<M.
  \end{split}
\end{equation}
In this representation the anti-symmetric Hamiltonian $H_A=\sum_j b_j (h_j-\tilde h_j)$ is almost of the same form as
the one given by the standard form \eqref{XYH}. The only difference is in the special boundary terms. We will return to
the Hamiltonian given by these special generators in a follow up paper.  

In the case of the FFD algebra we obtain
\begin{equation}
  \label{FFDdefrep}
  \begin{split}
    h_1&=X_1,\quad h_2=Z_1X_2,\\
    &\qquad\quad h_j=Z_{j-2}Z_{j-1}X_j,\quad j>2   \\
    \tilde    h_M&=X_M,\quad\tilde h_{M-1}=X_{M-1}Z_M\\
    &\qquad\quad \tilde h_j=X_j Z_{j+1}Z_{j+2},\quad j<M-1.
  \end{split}
\end{equation}
This is almost identical to the representation given by \eqref{bulkrep1}, with the only difference being in the boundary terms.

\bigskip

Now we discuss the completeness of the defining representation. The full operator algebra of the spin chain has
dimension $4^M$, and the tensor product $\AAA\otimes\AAA$ also has a total dimension $4^M$. Therefore, the two algebras
could be isomorphic.
However, the
representation might not span the full operator algebra of the chain. This can happen if the representation is not
faithful, i.e. certain elements of the tensor product are mapped to the identity.

\begin{thm}
  If the determinant of $A$ is even, then the defining representation is not faithful.
\end{thm}
\begin{proof}
If the determinant is even, then there is a non-trivial central element in the algebra.
This element is central also when embedded into the tensor product,
therefore it commutes with every element in the defining
representation. The representation of that central element can not be proportional to the identity, because the $X$
operators that appear in the representant act on different sites and the $Z$ operators that appear can not cancel
them. If the defining representation was complete, then we would have a non-trivial central element in the full
operator algebra of the chain. However, the only central elements in the operator algebra are the scalars. Therefore, the
representation can not be complete, and there has to be at least one non-trivial operator which is not in the
representation. This means that the representation has a dimension that is smaller than $4^M$, therefore it is not
faithful.
\end{proof}

\begin{thm}
  If the determinant of $A$ is odd, then the defining representation is faithful.
\end{thm}
\begin{proof}
  Let us assume that there is an element in the tensor product $\AAA\otimes\AAA$, which is represented by a scalar. This can only happen
  if it is of the form
\begin{equation}
(h_{a_1}\tilde h_{a_1})  (h_{a_2}\tilde h_{a_2})\dots  (h_{a_n}\tilde h_{a_n}),
\end{equation}
because this is the only way to cancel the $X$ operators in the representant. This element has to be central in the
tensor product algebra, otherwise it could not be represented by a scalar. This also means that $h_{a_1}h_{a_2}\dots
h_{a_n}$ is central in $\AAA$. However, this can only happen if $\det(A)$ is even. Therefore we proved the theorem
by contradiction.
\end{proof}

The computation of $\det(A)$ is very direct and convenient, especially for small system sizes. However, for larger
systems a direct proof of the non-degeneracy of the matrix $A$ might be easier.

Let us now discuss a central implication of the operator-state correspondence. We consider a Hamiltonian  $H=\sum_{j=1}^M
b_jh_j$, and the anti-symmetric combination
\begin{equation}
  \label{HA1}
  H_A=\sum_{j=1}^M b_j (h_j-\tilde h_j)
\end{equation}
in the defining representation.

\begin{thm}
  Every operator in the abstract algebra $\AAA$ which commutes with $H$ is mapped to a null-vector of $H_A$, where it is
  understood that $h_j$ and $\tilde h_j$ of \eqref{HA1} are represented by \eqref{leftm}-\eqref{rightm}.
\end{thm}
\begin{proof}
  Left- and right-multiplication in the algebra is represented by formulas \eqref{leftm}-\eqref{rightm}, therefore the
  commutator with $H$ is represented by \eqref{HA1}. If the commutator is zero, then we get a null vector.
\end{proof}

A special case of this theorem is when we choose the identity operator. It is mapped to the reference state:
\begin{equation}
  \varphi(1)=\ket{\emptyset}\equiv\ket{000\dots 0}.
\end{equation}
It follows from the theorem that the
reference state is a null-vector of $H_A$ in the defining representation. This is the algebraic reason as to why we
find suitable reference states in our concrete examples, presented below.

\subsection{\label{subsec:edge}The edge operator}

The construction of the free fermionic operators requires an edge operator \cite{ffd,fermions-behind-the-disguise}.
Its definition is that it commutes with all generators except with those that belong to a selected simplicial clique in
the frustration graph, with which it anti-commutes. This means that the edge operator $\chi$ has a commutation pattern,
which consists of zeroes, except for those vertices that belong to the simplicial clique. Now we discuss whether $\chi$
can be chosen to lie within the algebra itself, and its possible representations within the defining representation or
its extensions.

As an abstract operator, the edge operator might exist within the original graph-Clifford algebra itself, or it might
need an extension of the 
algebra. There are two ways to determine which case holds. First, one could check the determinant of the adjacency
matrix. If it is an odd number, then it follows from Theorem \ref{thm:non-degenerate}  that every commutation pattern
is present, therefore the edge operator can be chosen as one of the basis elements. However, if the determinant is even,
then further computations are needed, because the edge operator might still exist within the algebra.
In such a case
we need
to check whether its commutation pattern is in the span of the rows of $A$. This is a linear algebra problem over $\Zk$,
which can be addressed by direct computation, by a Gauss elimination over $\Zk$.

Irrespective of whether $\chi$ is a member of the original algebra, it can always be represented in the defining
representation as
\begin{align}
  \label{eq:edgeop-in-rep}
  \chi = \prod_{j \in K} Z_{j},
\end{align}
where $K$ denotes the simplicial clique.

Alternatively, we can
also 
enlarge the defining representation
so that $\chi$ can be accommodated \cite{ffd,fermions-behind-the-disguise}. There are multiple ways of proceeding. We
choose a convention where we add a new qubit with index 0 to the model, and represent $\chi$ by
\begin{equation}
  \label{edgeadd1}
  \chi=X_0.
\end{equation}
At the same time, we also modify the defining representation of the generators as follows. We keep the representation of
$\tilde h_k$, but 
we represent $h_k$ by
\begin{equation}
  \label{edgeadd2}
  h_k=X_k(Z_0)^{s_k}\mathop{\prod_{j<k}}_{A_{j,k}=1}Z_j,
\end{equation}
where $s_k=1$ if $k$ is in the simplicial clique, and $s_k=0$ otherwise. This representation clearly satisfies the
desired commutation relations.

\subsection{Norms}

\label{sec:norms}

It is useful to introduce a scalar product and norm in the abstract algebra $\AAA$ as follows. We assign a norm equal to unity
for every basis element \eqref{basiselement}, and we say that two different basis elements are orthogonal to each other.
This scalar product is naturally conserved by the mapping to the auxiliary spin chain.

Let us also consider the defining representation of a single copy of the algebra $\AAA$.
In the defining representation we will use the Hilbert-Schmidt norm
\begin{equation}
  \HS{\ordo_1}{\ordo_2}=\frac{1}{2^M} \text{Tr}(\ordo_1^\dagger \ordo_2).
\end{equation}

\begin{thm}
  \label{thm:normcons}
  The operator-state correspondence conserves the norm.
\end{thm}
\begin{proof}
  In the defining representation the scalar products are computed via  products of Pauli strings. Every Pauli string has
  Hilbert-Schmidt norm equal to unity.
Furthermore, if we compute the scalar product of two non-identical operators, then we end up with a non-trivial Pauli string.
  Every
  non-identical Pauli string has zero trace, and this proves the theorem.
\end{proof}

\subsection{The example of the Ising/XY algebra}

\label{sec:isingxy}

Now we discuss the simple and physically relevant example of the Ising/XY algebra. The more complicated example of the
FFD model will be
treated later in Section \ref{sec:ffd}.

\bigskip

First we consider the Ising/XY algebra with open boundary conditions. The non-zero matrix elements of the adjacency matrix are
\begin{equation}
    A_{j,j+1}=A_{j+1,j}=1,\qquad j=1,\dots,M-1.
\end{equation}
Direct inspection shows that if $M=2k$, then $\det(A)=1$. Therefore, the algebra is non-degenerate. In contrast, for
$M=2k+1$ we find $\det(A)=0$. There is one central basis element in the algebra, namely $h_1h_3\dots h_{M}$.

The defining representation of this algebra is given by operators acting on a spin chain with length $M$, given by
\begin{equation}
  h_j=Z_{j-1}X_j,\qquad \tilde h_j=X_jZ_{j+1},
\end{equation}
where it is understood that $Z_0=Z_{M+1}=1$. The representation leading to the Dzyaloshinskii–Moriya interaction term
(see eq. \eqref{XYH}) is obtained after a global rotation, see  \eqref{eq:XYrot}.

If $M$ is odd, then the same algebra can be represented as the terms of the quantum Ising Hamiltonian. One possible
choice is given by
\begin{equation}
  h_{2j-1}=Z_j,\qquad h_{2j}=X_{j}X_{j+1}.
\end{equation}
In this representation the central element becomes
\begin{equation}
  \ZZ=\prod_{j=1}^M Z_j.
\end{equation}
This operator is usually interpreted as the fermionic parity. 

This representation does not generate the full operator algebra of the spin chain with length $[M/2]$. It is easy to see
that for example the operator $X_1$ can not be generated as products of the $h_j$. However, including this
single operator we can indeed generate the full operator algebra.

\bigskip

In the case of periodic boundary conditions we have an additional edge in the graph, corresponding to adding
\begin{equation}
  A_{1,M}=A_{M,1}=1.
\end{equation}
to the adjacency matrix. The numerical value of the determinant depends on $M$ with a periodicity of four,
however, we always find even numbers. In accordance with this observation,  direct computation shows that there are
central elements in the
algebra.

If $M$ is even then the center is 4 dimensional and we can choose its generators as
\begin{equation}
  h_1h_3\dots h_{M-1},\qquad h_2h_4\dots h_{M}.
\end{equation}
If $M$ is odd then we find only one non-trivial central element, namely
\begin{equation}
  h_1h_2h_3\dots h_{M}.
\end{equation}

Let us finally consider the Ising spin chain representation of this algebra, with periodic boundary conditions. Then $M$
is even, the operators are represented by \eqref{IsingH}. We find that the central element $h_1h_3\dots h_{M-1}$ is
trivially mapped to the identity operator, whereas the other independent central element $h_2h_4\dots h_{M}$ is mapped
to $\ZZ$.

\section{Eigenstates}

In this Section we present our main results for the construction of the eigenstates of both $H_A$ and also the original
$H$. All results are discussed in a model independent way.
Many technical details depend on the concrete model;
in the case of the FFD model the details are discussed in Section \ref{sec:ffd}.

\subsection{Eigenstates of $H_A$}

\label{sec:eigenstates}

We consider models where we can embed two copies of the same algebra into the same representation. This means
that we have a model with two sets of generators given by $h_j$ and $\tilde h_j$.

It is our goal to derive eigenstates of the Hamiltonian
\begin{equation}
  \label{HAdef1}
  H_A=\sum_{j=1}^M   b_j(h_j-\tilde h_j)
\end{equation}
in a relatively simple way.

The usual construction of eigenstates in free fermionic
models uses a reference state, which is annihilated by one half of the fermionic eigenmodes. Afterwards states are created by
acting with the other half of the fermionic operators. This method can not  be used to find eigenstates of $H$, because
generally there is no reference state. This is true for the quantum Ising  chain, the anisotropic XY chain, the FFD
model, and many other models.

Our key insights is
that very often the anti-symmetric
combination $H_A$ given by \eqref{HAdef} does have a reference state $\refst$, which
is a product state
in the computational basis. The key relation is
\begin{equation}
  \label{reflocal}
  h_j\ket{\emptyset}=\tilde h_j \ket{\emptyset}, \qquad j=1,\dots,M,
\end{equation}
which can be established in multiple models.  If \eqref{reflocal} is satisfied, then
\begin{equation}
  H_A\ket{\emptyset}=0.
\end{equation}
It is crucial that the coupling constants in $H$ and $\tilde H$ are chosen to be the same.
However, the reference state does not need to be
translationally invariant: it can also have spatial modulations.

In the previous Section we showed, that if $H_A$ is given in the defining representation, then the reference state is
naturally given by
\begin{equation}
  \ket{\emptyset}=\ket{000\dots 0}.
\end{equation}
This holds true for both the XY model (after a trivial global rotation) and the FFD model.

Now we construct fermionic ladder operators that create the eigenstates of $H_A$. We have $H_A=H-\tilde H$, such that
$H$ and $\tilde H$ are free fermionic Hamiltonians that commute with each other. It is a natural idea to construct
fermionic operators for both Hamiltonians separately, using the formalism explained above in Section \ref{sec:fermion-intro}.

We use the formulas
\begin{equation}
  \begin{split}
  \label{Psidef2}
    \Psi_{\pm k}&=\frac{1}{\NN_k}T(\mp u_k) \chi T(\pm u_k)\\
   \tilde \Psi_{\pm k}&=\frac{1}{\NN_k}\tilde T(\mp u_k) \tilde \chi \tilde T(\pm u_k).\\
  \end{split}
\end{equation}
Here $T(u)$ and $\tilde T(u)$ are transfer matrices for the Hamiltonians $H$ and $\tilde H$, respectively.
The two Hamiltonians share the coupling constants, therefore the transfer matrices satisfy the same inversion relations,
and the roots $u_k$ of the polynomials $P(u^2)$ also coincide.

The formulas above use two edge operators $\chi$ and $\tilde \chi$. They satisfy the required commutation relations
within the two copies of the algebra: they commute with almost all generators, except those in the selected simplicial
clique.
Furthermore, we require the commutation relations
\begin{equation}
  \label{jochi}
  [\chi,\tilde \chi]=[\chi,\tilde h_k]=[\tilde\chi,h_k]=0.
\end{equation}
This ensures the commutation relations
\begin{equation}
  \label{HPsi2}
  [H,\tilde \Psi_k]=[\tilde H,\Psi_k]=[\Psi_k,\tilde \Psi_\ell]=0
\end{equation}
leading to a convenient set of relations
\begin{equation}
  \label{HPsi3}
  \begin{split}
    [H_A,\Psi_k]&=-2\eps_k\Psi_k, \quad   [H_A,\tilde \Psi_k]=2\eps_k\tilde\Psi_k,     \\
     [H_A,\Psi_{-k}]&=2\eps_k\Psi_{-k}, \quad   [H_A,\tilde \Psi_{-k}]=-2\eps_k\tilde\Psi_{-k}.
  \end{split}
\end{equation}
In certain cases it might happen that the two edge operators with the desired properties can be naturally found within a
given representation. Alternatively, they can always be constructed by using additional qubits. One possible way
to do this was given in Subsection \ref{subsec:edge} by formulas \eqref{edgeadd1}-\eqref{edgeadd2}. In this construction
a total number of two extra qubits is needed, corresponding to the addition of $\chi$ and $\tilde\chi$.

We note that in the Hermitian cases when the coupling constants $b_j$ are all real we get the Hermitian conjugate
relations
\begin{equation}
  \Psi_{-k}=\Psi_k^\dagger,\qquad
   \tilde \Psi_{-k}=\tilde\Psi_k^\dagger.
\end{equation}
 
Having collected all ingredients, now we can act with the fermionic operators to create excitations on top of the
reference state. However, there are crucial differences as opposed to the standard situation.

Now the formulas \eqref{HPsi3}  imply that we could act with both families of
fermionic operators, and for each mode $k$ we can both lower and also raise the energy by an amount $\eps_k$. Indeed,
the states $\Psi_k\refst$ and $\tilde\Psi_{k}^\dagger\refst$ have energy $-2\eps_k$, whereas  $\Psi_{k}^\dagger\refst$ and
$\tilde\Psi_{k}\refst$ have energy $2\eps_k$. At this stage there is no indication that any of these states would be a
null-vector, and later we prove rigorously that they have a finite norm. Thus we can act in both directions, by
using both families of fermionic operators. Later we will see that one family can be discarded: for example, it is
sufficient to act with the family $\Psi_{\pm k}$ only.

This situation described above is in contrast with the usual case in known Jordan-Wigner solvable models with a reference state.
There the {\it annihilation operators} annihilate the reference state, and the
{\it creation operators} create the excitations above the reference state.

The reason for this unusual behavior is as follows. In either $H$ or $\tilde H$ the eigenmode $k$ leads to a pair of
single-particle eigenvalues $\pm \eps_k$. Taking the anti-symmetric combination $H_A$ the possibilities for a particular
eigenmode become 4-fold, given by $\pm \eps_k\pm \eps_k=2\eps_k,0,0,-2\eps_k$. Thus for each mode $k$ we get a
total number of 3 possible numerical values,
and the zero value
comes with a degeneracy  of 2. This implies that the Hamiltonian $H_A$ has in fact two copies of a fermionic doublet for each
$k$. Then it is very natural that from a level with zero energy we can go both  up and down in energy.

Taking for example the family $\Psi_k$ it is convenient to choose the basis states $\Psi_k\refst$ and {$\Psi_{k}^\dag\Psi_k\refst$} for one doublet, and $\Psi_{k}^\dagger\refst$ and $\Psi_{k}\Psi_{k}^\dagger\refst$ for the other doublet. It
follows from the fermionic algebra that both pairs of states form a representation of a doublet, see
Fig. \ref{fig:levels}. The original reference state is simply a linear combination of the two selected states
$\Psi_{k}^\dagger\Psi_k\refst$ and $\Psi_{k}\Psi_{k}^\dagger\refst$ at the zero level.
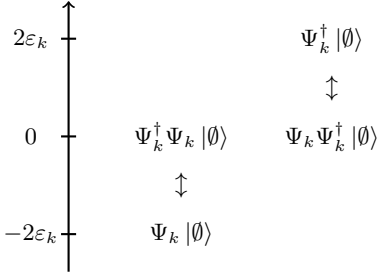
\begin{figure}[t]
  \centering

  \begin{tikzpicture}
    \node at (-1,0) {$\Psi_k^\dagger \Psi_k \ket{\emptyset}$};
    \node at (-1,-1.3) {$ \Psi_k \ket{\emptyset}$};
    \node at (-1,-0.65) {$\updownarrow$};
      \node at (1,0.65) {$\updownarrow$};
       \node at (1,0) {$\Psi_k \Psi_k^\dagger \ket{\emptyset}$};
       \node at (1,1.3) {$ \Psi_k^\dagger \ket{\emptyset}$};
       \draw[thick,->] (-2.5,-1.8) -- (-2.5,1.8);
       \draw[thick] (-2.6,1.3) -- (-2.4,1.3);
       \draw[thick] (-2.6,0) -- (-2.4,0);
       \draw[thick] (-2.6,-1.3) -- (-2.4,-1.3);
       \node at (-3,0) {0};
       \node at (-3,1.3) {$2\eps_k$};
         \node at (-3,-1.3) {$-2\eps_k$};
  \end{tikzpicture}

  \caption{The four states associated with the eigenmode with single-particle energy $\eps_k$. They split into two
    fermionic doublets, and the zero energy level is doubly occupied.}
  \label{fig:levels}
\end{figure}

The Hamiltonians $H$ and $\tilde H$ have typically $2^S$ different eigenvalues. If the coupling constants for $H$ and
$\tilde H$ would be allowed to
differ, then $H_A$ could have $4^S$ different eigenvalues. However, due to the reasons explained above, $H_A$ has a
total number of $3^S$ different eigenvalues. In this case the degeneracies of the different levels do differ, as opposed
to the case of $H$ that has homogeneous degeneracies.

Multi-particle states can be created by multiple action of the $\Psi_{k}$  operators. More concretely, the $n$-particle states
\begin{equation}
  \label{multip}
  \prod_{j=1}^n \Psi_{\sigma_j a_j}\refst
\end{equation}
with $a_j=1,\dots,S$, $j=1,\dots,n$, and $\sigma_j=\pm 1$
are eigenvectors of $H_A$ and they have eigenvalue
\begin{equation}
  E=-2\sum_{j=1}^n \sigma_j \eps_j.
\end{equation}
Further eigenstates with the same energy can be created by acting with the operator products
$\Psi_{\pm\ell}\Psi_{\mp\ell}$ for an index $\ell$ which is not present in the above product. In the case of the FFD
model we treat the completeness of this set of states in Section \ref{sec:allstates}.

The states \eqref{multip} are not normalized. However, we have the following:
\begin{restatable}{thm}{normthm}
  \label{thm:norm}
  The norm squared of the states \eqref{multip} is $2^{-n}$.
\end{restatable}
\begin{proof}
{In order to prove the theorem it is enough to show that the norm of such states is independent from the signs
$\sigma_j$. Once this is shown, we can perform a summation of the norm squared over all signs $\sigma_j$. The summation
leads to the end result being 1 due to the canonical anti-commutation relations  \eqref{Diracdef}, and if every summand
is equal then the summation is equivalent to multiplication by $2^n$, thus proving the theorem.}

{Now we use the result of Theorem \ref{thm:normcons}, and instead of the norms of the states we compute the
Hilbert-Schmidt norms in operator space. The 
  operator corresponding to the state
  \begin{equation}
    \prod_{j=1}^n \Psi_{\sigma_j k_j}\refst
  \end{equation}
  is simply just $\prod_{j=1}^n \Psi_{\sigma_j k_j}$. We now consider the
  Hilbert-Schmidt norm and after some re-ordering we are faced with the quantity
  \begin{equation}
\frac{1}{2^M}     \text{Tr}  \prod_{j=1}^n (\Psi_{-\sigma_j k_j} \Psi_{\sigma_j k_j}).
  \end{equation}
This is a product of projectors. We remind that the energies are of the form \eqref{energyformula}. Now
every factor $\Psi_{-\sigma_j k_j} \Psi_{\sigma_j k_j}$ projects to the sector of the
Hilbert space where the sign of $\eps_{k_j}$ in the energy formula is $\sigma_j$. This implies
that the operator product above projects to a particular sector of the Hilbert space where a total number of $n$ signs
are fixed, but $S-n$ signs are arbitrary. In the case of a projector the trace simply just counts the dimensionality of
the image space. It was proven in \cite{ffd} that every eigenenergy has the degeneracy
$2^{M-S}$. This means that the image space has dimensionality $2^{S-n}2^{M-S}=2^{M-n}$. Combining this with the overall
normalization of the trace we obtain the statement of the theorem.}
\end{proof}

So far we acted only with one family of fermionic operators, and we obtained eigenstates for every possible level.
In concrete cases we can show that it is not necessary to involve the other fermionic family $\tilde \Psi_k$ in the
process of creating eigenstates. This happens because often we can establish the relation
\begin{equation}
  \label{psitildepsi1}
\Psi_{\pm k}\refst=\tilde\Psi_{\mp k}\refst,
\end{equation}
making the use of the family $\tilde \Psi_k$ unnecessary. The  requirements for proving
\eqref{psitildepsi1} depend on the model and the representation chosen, and it might not always be possible to find such
a simple and compact relation.
{However, if the frustration graph $G$ is even-hole- and claw-free,  and the edge operator can be embedded into the
original algebra $\AAA$, then \eqref{psitildepsi1} can be established. We will treat this in a follow up paper.}

So far we have not discussed the full degeneracy of the various levels. We have only claimed that we can use the
fermions to create at least one eigenstate for every possible energy level of $H_A$. In order to create all possible
eigenstates we might need to deal with a large non-abelian symmetry algebra, which is often present in such models. In
the case of the FFD model this algebra was constructed recently in \cite{eric-lorenzo-ffd-1}; we
discuss it in Section \ref{sec:allstates}. The symmetry algebra is not known more generally.

\subsection{Eigenstates of $H$}

Above we constructed the eigenstates of the anti-symmetric Hamiltonian $H_A$. The reader might wonder whether we could
also find the eigenvectors of $H$ and $\tilde H$, and whether some of the eigenvectors of $H_A$ that we found could be
eigenvectors of $H$. Unfortunately, the solution is not obvious, because in most cases
 the reference state $\refst$ is not an eigenvector of $H$. We remind that even though $H$, $\tilde H$ and $H_A$ all
 commute with
 each other, these operators have different degeneracies, therefore it is not guaranteed
that a certain eigenvector of $H_A$ is also an eigenvector of $H$ or $\tilde H$.

Despite these difficulties, we can construct the eigenvectors of $H$, although in a less convenient way. The key
idea is that if we take a
random vector and act with either a raising or a lowering operator {\it for every mode}, then we necessarily get an
eigenvector, unless we end up with a null vector. This can give at least one eigenvector for every level
of $H$. Further eigenstates can then be obtained by using the hidden non-abelian symmetries, that can be present
in such models  \cite{eric-lorenzo-ffd-1}.

Instead of choosing a random vector, we start this
procedure from the reference state. Our key result here is that doing this we never  end up at a null vector, therefore
this procedure in fact yields at least one eigenvector for every energy level of $H$.

We take formula \eqref{multip} and act with $n=S$ fermionic modes:
\begin{equation}
  \label{multip2}
  \prod_{j=1}^S \Psi_{\sigma_j j}\refst.
\end{equation}
According to Theorem \ref{thm:norm} this vector has norm squared $2^{-S}$. It follows from the fermionic algebra that
it is an eigenvector of $H$ with eigenvalue 
\begin{equation}
  E=-2\sum_{j=1}^S \sigma_j \eps_j.
\end{equation}

It might appear that for practical purposes it is not convenient to act with all the eigenmodes to obtain an
eigenstate, especially if we were to use the formula \eqref{Psidef}, which might be computationally ineffective.
However, in the next Subsection we show that in certain cases there are factorized formulas available
for the fermionic operators. Such factorizations can render the formula \eqref{multip2}  efficient. In particular,
eigenstates of $H$ can be created on a classical computer with the action of $\ordo(M^2)$ localized
linear operations.

\subsection{The factorized formulas}

\label{sec:factorized}

The formula \eqref{Psidef} is rather general: it  describes the fermionic operators in a
large class of models covered
by the graph theoretical conditions of \cite{fermions-behind-the-disguise,unified-graph-th}.
At the same time, for certain models we can find other representations too.

In the one dimensional models the transfer matrix \eqref{Tudef} can often be represented as a Matrix Product Operator
(MPO). If the MPO has bond dimension $\kappa$, then the fermionic operators of \eqref{Psidef} can also be expressed as
an MPO, but with bond dimension $\kappa^2$. Depending on the model it might be possible to lower this bond
dimension. For the case of the FFD model we summarize the MPO construction of \cite{ffd} in Appendix \ref{sec:MPOT}.

Alternatively, in certain models the transfer matrix can also be factorized into a product of localized operators. This
happens naturally in the case of open boundary conditions, where the previously mentioned MPO is naturally expressed as
an operator product. This was established for the FFD model in \cite{ffd}, and more generally it can be established for
the so-called {\it chordal graphs} \cite{adrian-unpublished}.

Now we restate that result using the original FFD model~\cite{fendley-fermions-in-disguise} as an example, and we use
it to find a product form for the fermionic operators. 

The factorized formula of the transfer matrix takes the form
\begin{equation}
T(u)=G_M^T(u)G_M(u).\label{eq:TGG}
\end{equation}
Here
\begin{align}
	G_{M}(u)& \equiv g_{1}(u)g_{2}(u)\ldots g_{M}(u), \nonumber \\
	G_{M}^{T}(u) &=g_{M}(u)g_{M-1}(u)\ldots g_{1}(u),\label{eq:GG}
\end{align}
where the factors $g_{m}(u)$ are localized operators given by
\begin{equation}
g_{m}(u)=\cos\frac{\phi_{m}}{2}+\sin\frac{\phi_{m}}{2}h_{m}\label{eq:gm}
\end{equation}
and $[\cdot]^{T}$ means the usual transpose operation, assuming also $h_k^T=h_k$. The angles $\phi_m$ depend on the
spectral parameter and the
coupling strengths $b_m$ of the Hamiltonian. The relation is model-dependent, and it will be studied separately in later
Sections. In the case of more general chordal graphs, a special ordering of the $h_j$ or $g_j$ is needed
\cite{adrian-unpublished}.

From (\ref{eq:GG}) and \eqref{Psidef} it is already clear that the fermionic operators have a product form:
each $\Psi_{\pm k}$  is a product of $4M$ localized $g_m$ operators together with the insertion of the edge operator. In
certain cases we can simplify this picture even further. The edge operator commutes with most of the generators $h_k$,
therefore it also commutes with most of the $g_k$.
Therefore, we may present the fermionic operators in a more economic form as
\begin{align}
\label{eq:PsiGG}
  \Psi_{k} \sim
  g_{M}^{-}g_{M-1}^{-}\ldots g_{2}^{-}g_{1}^{-}\chi  g_{1}^{+}g_{2}^{+}\ldots g_{M-1}^{+}g_{M}^{+},
\end{align}
where we used $g_{m}^{\pm}\equiv g_{m}(\pm u_k)$
for brevity (cf. Appendix B.2 of \cite{ffd}). This formula gives $\Psi_k$ as a product of $2M$ localized operators, with
the insertion of the edge operator. 

Such a formula can be established in the Ising and XY models, and also the FFD model.

\section{The FFD model}

\label{sec:ffd}

In this Section we study the FFD model in detail. First we discuss the structure of the abstract algebra, and establish
statements about the completeness of the defining representation. Afterwards we turn our attention to the eigenstates of
$H_A$.

As an addition to all of this, in Appendix \ref{sec:periodic} we also present a few selected eigenstates of $H_A$ in the
case of periodic boundary conditions.

\subsection{Model and transfer matrix}

The FFD algebra is defined in 
eq. \eqref{ffdalg}, and the corresponding  frustration graph is depicted on Fig \ref{fig:ffd}. We restrict ourselves to open boundary
conditions, where there is no periodicity in the commutation relations or in the graph. The abstract  algebra will be
denoted as $\ffda$.

We also consider the second copy of the same algebra, denoted by $\tffda$, which is generated by the elements $\tilde
h_j$, with $j=1, 2,\dots, M$. In the following it will be convenient to consider the product algebra
$\ffda\otimes\tffda$.

In this Section we will use the defining  representation of the algebra, given
by \eqref{FFDdefrep}.

We first discuss the particular properties of the transfer matrix.
We denote the length of the spin chain explicitly, so that $T_M(u)$ is the transfer
matrix associated to the chain with length $M$.
There is a simple recursion relation:
\begin{equation}
  \label{Trecurs}
  T_M(u)=T_{M-1}(u)-u {\mathtt H}_M T_{M-3}(u)\,,
\end{equation}
with the initial conditions $T_M(u)=1$ for $M\le 0$.

Alternatively,  $T_M(u)$ can be given via a Matrix Product Operator, see \cite{ffd} and Appendix
\ref{sec:MPOT}. There is also a factorized form \cite{ffd}, which can be used to construct quantum circuits
\cite{sajat-floquet}, but we do not use it here.

The transfer matrix satisfies the inversion relation \eqref{eq:inversion}, with a polynomial $P_M(x)$ that
is given by a recursion relation analogous to \eqref{Trecurs}
\begin{equation}
  \label{Precursion}
  P_M(x)=P_{M-1}(x)-x b_M^2 P_{M-3}(x)\,.
\end{equation}
with the initial condition $P_M(x)=1$ for $M\le 0$.

\subsection{The algebra}

Now we discuss the abstract FFD algebra and its defining representation.

The abstract FFD algebras $\ffda$ and $\tffda$ have dimension $2^M$, therefore the abstract tensor product
$\ffda\otimes\tffda$ has dimension $4^M$. This coincides with the dimensionality of the full operator algebra of the
spin chain with length $M$. Therefore, there is a chance that the defining representation of $\ffda\otimes\tffda$ covers
the full operator algebra. However, this happens only for particular values of $M$. Now we study the problem using the
tools introduced in Section \ref{sec:graph-clifford}.
Elements of such computations can be found in
\cite{sajat-floquet,eric-lorenzo-ffd-1},
but the complete exploration of this issue appears to be new.

The non-zero elements of the adjacency matrix are
\begin{equation}
  \begin{split}
    A_{j,j+1}&=A_{j+1,j}=1,\qquad j=1,\dots,M-1\\
    A_{j,j+2}&=A_{j+2,j}=1,\qquad j=1,\dots,M-2.
  \end{split}
\end{equation}

The determinant depends on the remainder of $M$ modulo 6. The values of the determinant are listed in Table
\ref{tab:ffd}. That Table also lists the number of independent non-trivial central elements $n_c$ in the algebra. We do
not prove the values for the determinants, but the values for $n_c$ are proven rigorously.

\begin{table}[h!]
  \centering
  \begin{tabular}{|c||c|c|c|c|c|c|}
    \hline
    $M$    &$6k$ & $6k+1$ & $6k+2$ &$6k+3$&$6k+4$&$6k+5$\\
    \hline
    $\det(A)$ & $2k+1$ & 0 & $-(2k+1)$ & $2k+2$ & $0$ & $-(2k+2)$ \\
    \hline
    $n_c$ & 0 & 1 & 0 &1 &2 &1 \\
    \hline
  \end{tabular}
  \caption{Structural properties of the FFD algebra: The determinant of the adjacency matrix, and the number of
    independent non-trivial central elements.}
  \label{tab:ffd}
\end{table}

\begin{thm}
  The FFD algebra with $M$ generators has $n_c$ number of independent non-trivial central elements, where $n_c$ is given
  in Table 
  \ref{tab:ffd}.
  \label{thm:ffd}
\end{thm}
\begin{proof}
We compute the rank of the adjacency matrix over $\Zk$. We have a Toeplitz matrix:
\begin{equation}
A=  \begin{pmatrix}
    0 & 1 & 1 & 0& 0& 0 & 0&  \dots\\
    1 & 0 & 1 & 1 & 0&  0 &0&  \\
    1& 1 & 0 & 1 & 1 & 0 &0 & \\
    0& 1 & 1 & 0 & 1 & 1 & 0 &\\
 0&      0& 1 & 1 & 0 & 1 & 1  &\\
\vdots & & & & & & &\ddots \\
  \end{pmatrix}
\end{equation}
This matrix is not of full rank, if there are some rows that are linearly dependent. Working over $\Zk$ this means that
the sum of certain vectors is a null vector over $\Zk$. If $v_j$ denote the rows of the matrix, we would have at least
one relation like
\begin{equation}
  v_{a_1}+v_{a_2}+\dots+v_{a_n}=0\quad \text{mod } 2,
\end{equation}
where we can assume that $a_1<a_2<\dots<a_n$.

It is easy to see that $a_1\le 3$ and $a_n\ge M-2$. If this were not the case, then the sum of the vectors would have
non-vanishing elements at the left-most position of $v_{a_1}$ and the right most position of $v_{a_n}$. So we need to
investigate the possibilities of having a linear relation like above with $a_1=1, 2, 3$. We do it case by case.

If $a_1=1$, then either $v_3$ or $v_4$ has to be in the sum, based on column 2, but both can not be in the sum
simultaneously. If $v_3$ is in the sum, then $v_2$ has to be there too, based on column 1. In this case we get the
partial sum
\begin{equation}
  v_1+v_2+v_3=(0,0,0,0,1,0,0,\dots,0).
\end{equation}
It follows that $v_4$, $v_5$ and $v_6$ can not be in the sum, but $v_7$ has to be there. Continuing these steps, we can
eventually show that the pattern has to be
\begin{equation}
  \label{pat2}
  (v_1+v_2+v_3)+(v_7+v_8+v_9)+\dots
\end{equation}
and the pattern can be completed if $M=6k+3$.

The next case is when $a_1=1$, so $v_1$ is in the sum, but $v_3$ is not, therefore $v_4$ is there. Then the first
partial sum we can compute is
\begin{equation}
  v_1+v_4=(0,0,0,0,1,1,0,\dots,0).
\end{equation}
This implies that $v_5$ and $v_6$ can not be in the sum, but $v_7$ needs to be there. Eventually we can complete the
pattern, to find the relation
\begin{equation}
  \label{pat1}
  v_1+v_4+v_7+\dots=0,
\end{equation}
which can be completed if $M=3k+1$.

The next case is when $a_1>1$. We argued that $a_1=2$ or $a_1=3$, but based on column 1 both $v_2$ and $v_3$ need to be
present in the sum. Based on column 2 also $v_4$ needs to be present. Then we get the partial sum
\begin{equation}
  v_2+v_3+v_4=(0,0,0,0,0,1,0,\dots,0).
\end{equation}
Based on columns 3, 4 and 5 we see that $v_5$, $v_6$, $v_7$ can not be present, but later $v_8$ needs to be
present. Continuing this line of reasoning we see the emergence of the pattern
\begin{equation}
  \label{pat3}
   (v_2+v_3+v_4)+(v_8+v_9+v_{10})+\dots
\end{equation}
The pattern can be completed if $M=6k+5$.
\end{proof}

This Theorem and the derivations in Section \ref{sec:graph-clifford} imply that 
for $M=6k$ and $M=6k+2$ the defining representation of the product of the two FFD algebras $\ffda\otimes\tffda$
coincides with the full operator algebra of the spin chain with length $M$.  For other values of $M$ the
center of the FFD algebra is non-trivial, and the defining representation does not cover the full operator algebra.

For the values $M=6k+j$ with $j=1,3,4,5$ we list the non-trivial central elements, and we also show that the defining
representation of $\ffda\otimes\tffda$ is not faithful.

For $M=3k+1$ we find one non-trivial element in the center of the algebras, corresponding to the pattern \eqref{pat1}.
In the defining representation this element
belongs indeed to both algebras; we find the relation:
\begin{equation}
  \label{relx3}
  \begin{split}
    &  h_1h_4h_7\dots h_{M}= \tilde h_1\tilde h_4\tilde h_7\dots \tilde h_{M}\\
  &  =X_1Z_2Z_3X_4Z_5Z_6\dots Z_{M-2}Z_{M-1}X_M.
  \end{split}
\end{equation}
Direct computation confirms, that the corresponding element  is central in $\ffda$ and $\tffda$.

Similarly, for $M=6k+3$ the pattern \eqref{pat2} gives the central elements 
\begin{equation}
  \label{k3id}
  \begin{split}
&      (h_1h_2h_3) (h_7h_8h_9)\dots (h_{M-2}h_{M-1}h_M)=\\
      &=   (\tilde h_1\tilde h_2\tilde h_3) (\tilde h_7\tilde h_8\tilde h_9)\dots (\tilde h_{M-2}\tilde h_{M-1}\tilde h_M).
  \end{split}
\end{equation}

For $M=6k+5$ the pattern \eqref{pat3} gives the central elements
\begin{equation}
  \label{k5id}
  \begin{split}
&     (h_2h_3h_4) (h_8h_9h_{10})\dots (h_{M-3}h_{M-2}h_{M-1})=\\
&  (\tilde h_2\tilde h_3\tilde h_4) (\tilde h_8\tilde h_9\tilde h_{10})\dots (\tilde h_{M-3}\tilde h_{M-2}\tilde h_{M-1}).
  \end{split}
\end{equation}

For $M=6k+4$ we have the central elements and also another independent relation, which comes from
other central elements. We can choose for example
\begin{equation}
  \begin{split}
&  (h_2h_3h_4) (h_8h_9h_{10})\dots (h_{M-2}h_{M-1}h_{M})=\\
&  (\tilde h_2\tilde h_3\tilde h_4) (\tilde h_8\tilde h_9\tilde h_{10})\dots (\tilde h_{M-2}\tilde h_{M-1}\tilde h_{M}).
  \end{split}
\end{equation}

In these cases the defining representation can not cover the full operator algebra, because we have a coincidence
between representants of an element from $\ffda$ and $\tffda$.

\subsection{The edge operator}

\label{sec:edge}

It is our goal to choose the edge operator within the FFD algebra. More precisely, we intend to find the two copies of
the edge operators $\chi$ and $\tilde\chi$ in $\ffda$ and $\tffda$.
This would guarantee the relations \eqref{jochi},
which are needed for the mutual commutativity \eqref{HPsi2}, that eventually guarantees the raising and lowering
relations \eqref{HPsi3}.

We choose the simplicial clique as the set with the single element $h_1$, therefore the edge operator has to satisfy
\begin{equation}
  \{\chi,h_1\}=0,\qquad [\chi,h_k]=0,\quad k>1.
\end{equation}

Theorem \ref{thm:ffd} implies that in the special cases $M=6k$, $M=6k+2$
we can choose $\chi$ to lie within $\ffda$. The concrete expressions we find are as follows.

For $M=6k$ we can choose
\begin{equation}\label{eq:chi6k}
	\chi=i^k (h_2h_3h_4)(h_8h_9h_{10})\cdots (h_{6k-4}h_{6k-3}h_{6k-2}).
   \end{equation}
For $M=6k+2$ we can choose
\begin{equation}
  \chi
    =h_1h_4h_7\cdots h_{6k+1}.
   \end{equation}

Interestingly, there is a solution also for $M=6k+5$, despite the algebra being degenerate. We find that the edge
operator can be chosen as
\begin{equation}
  \chi
  =(ih_1h_2h_3)(ih_7h_8h_{9})\cdots (ih_{6k+1}h_{6k+2}h_{6k+3}).
\end{equation}

In the remaining cases $M=6k+1$, $M=6k+3$, $M=6k+4$ it can be proven that the edge operator can not be chosen within
$\ffda$. The proof is similar to the strategy in the case of Theorem \ref{thm:ffd}, but applied to a concrete case: One
needs to check that the commutation pattern of the edge operator is not in the linear span of the rows of the adjacency
matrix.

\subsection{Eigenstates of $H_A$}

Now we construct the eigenstates of $H_A$. We limit ourselves to the cases when $\chi$ exists within the FFD algebra.

First we construct the second copy of the fermionic operators, using the generators of $\tffda$.
They will be denoted as $\tilde \Psi_k$. It follows from the general arguments in Section \ref{sec:eigenstates} that 
\begin{equation}
  [\Psi_k,\tilde \Psi_l]=0,
\end{equation}
which also implies the desired cross-commutation \eqref{HPsi2}.

The fermions $\tilde \Psi_k$ are not to be confused with the auxiliary fermions of \cite{eric-lorenzo-ffd-1}, which are responsible
for symmetry operations {\it within} the FFD Hamiltonian. The symmetry operations will be discussed below in the next
Subsection. 

Once all of these relations are established, we can act with the fermionic operators on the reference state, to create
the eigenstates of $H_A$. The spatial structure of the states we get this way will be studied in more detail in a
follow-up paper.

We put forward the following theorem:
\begin{thm}\label{thm:Psi-tildePsi-ket}
If we build the two families of fermionic operators using the same formulas, with the only difference that every $h_j$
is replaced by $\tilde h_j$, and if the edge operators $\chi$ and $\tilde\chi$ can be embedded into the algebra and they
are also built using the same rule, then the following holds:
  \begin{equation}\label{eq:Psi-tildePsi-ket}
    \Psi_{\pm k}\refst=\tilde\Psi_{\mp k}\refst
    \,,
  \end{equation}
\end{thm}
The theorem will be  proven in a more general form in the second paper of this series.

This theorem implies that it suffices to act with one family of fermionic operators, to obtain at least one state per level.
Additional symmetry transformations can be applied afterwards.

\subsection{Symmetry operations}

\label{sec:allstates}

The work \cite{eric-lorenzo-ffd-1} derived the full symmetry algebra for the Hamiltonian $H$, using the same
representation as in this paper. Building on the results of \cite{eric-lorenzo-ffd-1} we claim that now we can obtain every
eigenvector of $H_A$ using the standard fermions $\Psi_k$ and one family of certain auxiliary fermions, introduced in
\cite{eric-lorenzo-ffd-1}. 

The paper \cite{eric-lorenzo-ffd-1} showed that every eigenvector of $H$ can be obtained in three steps. First one needs
an initial eigenvector. Second, one can set the desired eigenvalue by acting with desired raising or lowering
operators. Third, one can act with the large symmetry algebra of the model. In \cite{eric-lorenzo-ffd-1}  this symmetry
algebra was factorized, one product given by the algebra of all those operators that commute with the $h_k$ and the edge
operator, the other product given by an algebra of certain auxiliary fermionic operators, denoted as $\Psi'_m$.

We adapt the results of \cite{eric-lorenzo-ffd-1}  to our situation: we present a slightly modified argument for the
completeness of the spectrum of $H_A$.

First we construct the auxiliary fermions $\Psi'_m$ as presented in \cite{eric-lorenzo-ffd-1}. It was shown there that
this is a family of $S'=\lfloor(M+2)/6 \rfloor$ fermions, which commute with the Hamiltonian. These fermions are part of
the symmetry algebra of the model.

In \cite{eric-lorenzo-ffd-1}  an important
ingredient for the introduction of the auxiliary fermions  was  another edge operator denoted by $\chi'$  (Section V.A there). This
Hermitian operator also squares to the identity, it anti-commutes with $h_2$ and commutes with every other $h_k$. We
show that in the cases $M=6k$, $M=6k+2$ and $M=6k+5$ the operator $\chi'$ can also be embedded
into the FFD algebra. This implies that the auxiliary fermions can be constructed manifestly within the FFD algebra
$\ffda$.

For $M=6k$ we can choose
\begin{equation}\label{eq:chiv6k}
     \chi'=i^k (h_4h_5h_6)(h_{10}h_{11}h_{12})\cdots (h_{6k-2}h_{6k-1}h_{6k}),
 \end{equation}
for $M=6k+2$
\begin{equation}
   \chi'=i^k h_1(h_5h_6h_7)(h_{11}h_{12}h_{13})\cdots (h_{6k-1}h_{6k}h_{6k+1}),
 \end{equation}
 and for $M=6k+5$ we can choose
\begin{equation}
    \chi'=(ih_2h_3h_4)(ih_8h_9h_{10})\cdots (ih_{6k+2}h_{6k+3}h_{6k+4}).
 \end{equation}

We note a technical detail, which is different in our choice for $\chi'$.
In \cite{eric-lorenzo-ffd-1} the two operators were chosen such that $[\chi,\chi']=0 $, and this eventually lead to the
commutation relation $[\Psi_k, \Psi'_{\ell}]=0$ for all $k,\ell$. We find the same relations for $M=6k+2$. However, for
$M=6k$ our choice satisfies $\{\chi, \chi' \} = 0$, which eventually implies
$\{ \Psi_k, \Psi'_{\ell}\}=0$ for all $k,\ell$.

We can also consider the Hamiltonian $\tilde H$ and construct the auxiliary fermions $\tilde \Psi'$ by adapting the
results of \cite{eric-lorenzo-ffd-1}. All steps need to be repeated, with the introduction of the fourth edge operator
$\tilde\chi'$, which is a member of $\tffda$. With this we constructed four families of fermionic operators, $\Psi_k$
and $\tilde \Psi_k$ governing the spectrum of  
$H$ and $\tilde H$ separately, and $\Psi_k'$ and $\tilde \Psi_k'$ describing the symmetries of these Hamiltonians within
the $\ffda$ and $\tffda$ algebras, respectively. For a summary of the different fermionic operators see Figure \ref{tab:fermions}.

\begin{figure}[t]
  \centering
  \begin{tikzpicture}
    \node at (-2,0.75) {In $\ffda$:};
      \node at (2,0.75) {In $\tffda$:};
    \node at (-3,0) {  $\Psi_k$ };
    \node at (-1,0) {  $\Psi'_k$ };
    \node at (1,0) {  $\tilde\Psi_k$ };
    \node at (3,0) {  $\tilde\Psi'_k$ };
    \node[align=center] at (-3,-0.75) {spectrum \\ generating};
    \node[align=center] at (1,-0.75) {spectrum \\ generating};
    \node[align=center] at (-1,-0.75) {auxiliary};
    \node[align=center] at (3,-0.75) {auxiliary};
  \end{tikzpicture}

  \caption{Short summary of the notations for the different families of fermions. In the ``spectrum generating'' families
    the positive indices take values $k=1\dots S$, whereas for the auxiliary fermions we have $k=1\dots S'$. }
  \label{tab:fermions}
\end{figure}

We present a new conjecture: We claim that for $M=6k$ and $M=6k+2$ all eigenstates of $H_A$ can be generated by only two
families of fermions: the spectrum generating $\Psi_{\pm k}$, with $k=1,\dots,S$ and the auxiliary $\Psi_{\pm \ell}'$,
with $\ell=1,\dots,S'$, and no other symmetry operators are needed.

In order to formalize the conjecture, we introduce some notations.
Let us choose integers  $\kappa_j$, $\kappa_{-j}$, $j=1,\dots,S$, and also $\kappa_k'$, $\kappa_{-k}'$, $j=1,\dots,S'$,
all of which take values in $\mathbb{Z}_2$.  In other words, we introduce $2(S+S')$ classical bits.

Consider the family of states
\begin{equation}
  \label{completefamily}
   \left( \prod_{j=1}^{S'} (\Psi_{k}')^{\kappa'_{k}}(\Psi'_{-k})^{\kappa'_{-k}} \right)
  \left( \prod_{j=1}^S (\Psi_{k})^{\kappa_{k}}(\Psi_{-k})^{\kappa_{-k}} \right)
  \refst.
\end{equation}
In plain terms, the classical bits control whether or not the corresponding fermionic operator is used to act on the
reference state.

\begin{conje}
 If $M=6k$ and $M=6k+2$, then the family of states \eqref{completefamily} spans the whole Hilbert space in the defining
 representation.
\end{conje}

We did not manage to prove the conjecture, however, it was numerically observed to be true for the system sizes $M=6, 8,
12$, for randomly chosen couplings.

The counting of the dimensionality clearly supports the claim: the naive counting of the family of states
\eqref{completefamily} gives $4^{S+S'}$ different states, and using previous results for $S$ and $S'$ we find that this
number correctly reproduces the full Hilbert space dimension $2^M$. It also follows from the fermionic algebra, that two
states in the family, which have finite norm, can not be equal to each other. However, we can not prove that every
member in \eqref{completefamily} has a finite norm, because we lack the analogue of Theorem \ref{thm:norm}.

\subsection{Entanglement}

It is an interesting question, what is the entanglement content of the states that we created. 
Currently there is no exact method available to perform such computations. 

However, we can show that the action of every raising/lowering operator can add only an $\ordo(1)$ contribution
to the bi-partite spatial entanglement. This follows from the fact that the $\Psi_k$ can be written as a Matrix Product
Operator (MPO) with bond dimension 9. Such an MPO representation can be derived using the MPO representation of the
transfer matrix, originally derived in \cite{ffd}. Details are presented in Appendix \ref{sec:MPOT}.

Therefore, the single-particle and few-particle states have satisfy an area law in the $M\to\infty$ limit, if we keep
the number of excitations fixed.

\section{Discussion}

In this paper we discussed two closely related issues. First, we treated the algebras behind these models, by
establishing connections to the mathematical theory of graph-Clifford algebras. We also introduced the ``defining
representation'' of the tensor product of two identical graph-Clifford algebras; to out best knowledge this is a new
result. This representation naturally reproduces prominent models, such as the XY model in a particular presentation,
and the FFD model of Fendley.

As a second direction, we proposed to study the
anti-symmetric combination of Hamiltonians $H_A=H-\tilde H$. We showed that for this combination we can create few-body
eigenstates using the free fermionic operators. Afterwards we discussed various technical details in the case of the FFD
model. We also studied the problem of how to create all eigenstates of $H_A$ using two families of  fermions: the
spectrum generating family and the auxiliary family introduced in \cite{eric-lorenzo-ffd-1}.

In follow up papers we will introduce new explicit formulas for the fermionic operators, which will help us study their
physical properties, such as the distribution of the operator weight along the spin chain. Furthermore, we will show how
to connect this formalism with the traditional methods of Jordan-Wigner solvability.

In future work it would be interesting to apply the graph-Clifford algebra framework also to other models, which don't
have the free 
fermionic structures. In Appendix \ref{sec:periodic} we treated the FFD model with periodic boundary conditions, and
derived a few selected exact eigenstates. It would be interesting to consider $H_A$ in non-integrable model as well. 
If we consider $H_A$ in the defining representation, then the reference states are
exact null vectors. Such models do not have $U(1)$ invariance, and the existence of the exact null vector follows from
the particular anti-symmetry of $H_A$. It is an interesting question, whether other exact eigenstates could be found for
such a more general class of models.

We hope to return to these questions in future work.

\begin{acknowledgments}
We are thankful to P\'eter Csikv\'ari for useful discussions. 
B.P. and I.V. were supported by the Hungarian National Research,
Development and Innovation Office, NKFIH Grant No. K-145904 and B.P. was supported by the NKFIH excellence grant
TKP2021\_NKTA\_64.
K.F. was supported by MEXT KAKENHI Grant-in-Aid for Transformative Research Areas A “Extreme Universe” (KAKENHI Grant No. JP21H05191) and KAKENHI Grant No. JP25K23354 from the Japan Society for the Promotion of Science (JSPS).
\end{acknowledgments}

\appendix

\section{An MPO for the fermionic operators}

\label{sec:MPOT}

For the FFD model, the transfer matrix has an MPO representation with
bond dimension $\kappa=3$, applicable to both open and closed chains \cite{ffd}. 

We define the (operator valued) $3\times3$ matrix
\begin{equation}
T_{t_{0},t_{L}}(u)=\sum_{\{a_{m}\}_{m=1}^{L},\{t_{m}\}_{m=1}^{L-1}}\prod_{i=1}^{L}A_{t_{i-1}t_{i}}^{a_{i}}(ub_{i})\sigma_{i}^{a_{i}},\label{eq:TMPO}
\end{equation}
where $\sigma_{i}^{a}$ is $\mathbb{1},X_{i},Z_{i}$ for $a=0,x,z$
respectively. The sums for the physical and auxiliary indices run
over $a_{m},t_{m}=0,x,z$, and the uncontracted ones $t_{0},t_{L}$
can pick the same values. The tensors $A_{t',t}^{a}(u)$ have the
following non-vanishing components:
\begin{align}
A_{00}^{0}(u) & =1, & A_{0z}^{z}(u) & =u, & A_{x0}^{x}(u) & =1, & A_{zx}^{x}(u) & =-1.
\end{align}
Note that for inhomogeneous couplings $b_{i}$ in the Hamiltonian \eqref{Hdef} the parameter $u$ in the tensors of (\ref{eq:TMPO}) acquires spatial dependence.
To arrive at the transfer matrix, we have to take the $0,0$ component
of (\ref{eq:TMPO}) in the case of open boundary conditions (with $b_{M+1}=b_{M+2}=0$ for consistency), while
for the periodic boundary conditions we have to take the trace:
\begin{equation}
T_{M}(u)=\begin{cases}
\quad\quad\ \ T_{0,0}(u) & \begin{array}{cc}
L=M+2 & \text{(OBC)}\end{array}\\
\sum\limits_{t=0,x,z}T_{t,t}(u) & \begin{array}{cc}
L=M & \quad\ \ \,\text{(PBC)}\end{array}
\end{cases}.
\end{equation}

The fermionic modes $\Psi_{\pm k}$ of the open chain are then proportional
to the operator $T_{0,0}(\mp u_{k})\chi T_{0,0}(\pm u_{k})$ with $L=M+2$.
Assuming $\chi$ is a Pauli string, the structure of this formula
is similar to (\ref{eq:TMPO}) itself, with the caveat that now the
dimension of the MPO is $\kappa^{2}=9$. For example, if we choose
$\chi=Z_{1}$ the fermions may be written as
\begin{align}
	\Psi_{\pm k} =\frac{1}{\mathcal{N}_{k}}& \sum_{\alpha_{1},\ldots,\alpha_{M}=0}^{3}\sum_{\tau_{1},\ldots,\tau_{M+1}=0}^{8} \label{eq:MPOpsi}  \\
	&\tilde{\mathcal{A}}_{0,\tau_{1}}^{\alpha_{1}}\mathcal{A}_{\tau_{1},\tau_{2}}^{\alpha_{2}}
	 \ldots\mathcal{A}_{\tau_{M+1},0}^{\alpha_{M+2}}\sigma_{1}^{\alpha_{1}}\sigma_{2}^{\alpha_{2}}\ldots\sigma_{M+2}^{\alpha_{M+2}}, \nonumber
\end{align}
where $\sigma_{i}^{\alpha}$ is the $\alpha^{\text{th}}$ Pauli matrix
acting on site $i$, including $\alpha=0$ for the identity. Here the non-zero elements of the tensors $\mathcal{A}_{\tau',\tau}^{\alpha}$
are
\begin{align}
\mathcal{A}_{0,0}^{0} & =\mathcal{A}_{4,0}^{0}=\mathcal{A}_{8,4}^{0}=-\mathcal{A}_{5,1}^{0}=-\mathcal{A}_{7,3}^{0} \nonumber \\
	& =\mathcal{A}_{1,0}^{1}=\mathcal{A}_{3,0}^{1}=-\mathcal{A}_{2,1}^{1}=-\mathcal{A}_{6,3}^{1}=1  \\
i\mathcal{A}_{3,2}^{2} & =i\mathcal{A}_{1,6}^{2}=-i\mathcal{A}_{6,5}^{2}=-i\mathcal{A}_{2,7}^{2} \nonumber \\
	&=\mathcal{A}_{0,2}^{3}=-\mathcal{A}_{0,6}^{3}=u \\
\mathcal{A}_{0,8}^{0} & =-u^{2},
\end{align}
and for inhomogeneous couplings it is understood that $\mathcal{A}_{\tau_{m-1},\tau_{m}}^{\alpha_{m}}$
has to be parameterized with $u=\pm u_{k}b_{m}$, where $\pm u_{k}$
is the root corresponding to the mode $\Psi_{\pm k}$. The first tensor has to be modified
due to the presence of the edge operator that acts on the first site,
such that
\begin{align}
\mathcal{\tilde{A}}_{\tau',\tau}^{0} & =\mathcal{A}_{\tau',\tau}^{3}, & \mathcal{\tilde{A}}_{\tau',\tau}^{1} & =i\theta_{\tau',\tau}\mathcal{A}_{\tau',\tau}^{2},\\
\mathcal{\tilde{A}}_{\tau',\tau}^{3} & =\theta_{\tau',\tau}\mathcal{A}_{\tau',\tau}^{0} & \mathcal{\tilde{A}}_{\tau',\tau}^{2} & =i\theta_{\tau',\tau}\mathcal{A}_{\tau',\tau}^{1},
\end{align}
where $\theta_{\tau',\tau}=1$ if the distance of $\tau$ and $\tau'$
(in the modulo $9$ sense) is smaller than $2$ and $\theta_{\tau',\tau}=-1$
otherwise.

\section{Selected eigenstates of the periodic anti-symmetric FFD model}

\label{sec:periodic}

In this Section we study the eigenstates of the anti-symmetric Hamiltonian
\begin{equation}
  \label{HA2}
  H_A=\sum_{j=1}^M (Z_jZ_{j+1}X_{j+2}-X_jZ_{j+1}Z_{j+2})
\end{equation}
with periodic boundary conditions. This model is integrable \cite{ffd}, however, its exact eigenstates have not yet been found.
We use the ideas developed in Section \ref{sec:opstate} to find a few selected exact eigenstates of the model.

Furthermore, we will also consider the models of
\cite{alcaraz-medium-fermion-1,alcaraz-medium-fermion-2,alcaraz-medium-fermion-3}. More concretely, we consider the
family of anti-symmetric Hamiltonians
\begin{equation}
  \label{HAN}
  H_A^{(N)}=\sum_{j=1}^M (h_j-\tilde h_j),
\end{equation}
where now
\begin{equation}
  \begin{split}
    h_j&=Z_{j}\dots Z_{j+N-1} X_{j+N}\\
    \tilde h_j&=X_jZ_{j+1}\dots Z_{j+N}.
  \end{split}
\end{equation}

The leading thought is the same: we use the operator-state correspondence to map symmetry operators of $H$ into
null-vectors
of $H_A$. However, the operator-state correspondence that we developed is not compatible with the periodic boundary
conditions. Instead, we will argue as follows. We will consider states with manifest translational invariance, where the
null-vector condition is proven via the operator-state correspondence
in the bulk, and then can be extended to the boundary  by translational invariance.

Our first candidate is derived from an unpublished alternative transfer matrix of the periodic Hamiltonian $H$. We do
not discuss that transfer matrix here, as it is not relevant to our main derivations. Instead we present the result,
which is surprisingly simple. Let $\ket{v}\in\complex^2$ be an arbitrary vector. Then
\begin{equation}
  \otimes_{j=1}^M \ket{v}
\end{equation}
is a null vector of $H_A$. In fact, any such vector is a null vector for every Hamiltonian $H_A^{(N)}$. This statement
can be proven via a direct computation, by induction over $N$.

We can be more concrete and parameterize the product state as
\begin{equation}
  \ket{\Psi(a)}=  \otimes_{j=1}^M
  \begin{pmatrix}
    1 \\ a
  \end{pmatrix}.
\end{equation}
We can now take derivatives with respect to $a$. The $n$th derivative gives the Dicke state $\ket{D^{(M)}_n}$ with $n$
down spins, and it
follows that all these vectors are null-vectors of $H_A^{(N)}$.

We can also use the structure of the conserved charges \eqref{Qalpha} to find new eigenstates of $H_A$ in the periodic
case. The idea is to construct a state from the operators \eqref{Qalpha} and to extend it to the periodic case. In the
open case we get a null-vector of $H_A$, but due to the locality of $H$ and  translational invariance the state has to
be a null-vector also in the periodic case. Now we will consider the Hamiltonian \eqref{HA2}, the extension to
\eqref{HAN} is straightforward.

For computational basis states with $n$ spin downs placed at sites $a_1,\dots,a_n$ we use the notation
\begin{equation}
  \ket{a_1,\dots,a_n}.
\end{equation}
Then the first new null-vector is
\begin{equation}
  \ket{\Psi_2}=\sum_{a_1+2<a_2} \ket{a_1,a_2},
\end{equation}
where periodic boundary conditions are understood in the constraint for the summation. Proceeding further we obtain
states
\begin{equation}
  \ket{\Psi_n}=\sum_{a_1,\dots,a_n}\ket{a_1,\dots,a_n},
\end{equation}
where the constraint is
\begin{equation}
  a_j+2<a_{j+1}
\end{equation}
and periodic boundary conditions are understood for applying the constraint. These states can be considered as a
particular type of ``constrained Dicke states''.

There are even more possibilities for constructing null-vectors of these models. One could take for example local
extensive charges of $H$ in the periodic case, and construct eigenstates for $H_A$ from them. The general structure of local
charges for $H$ is not yet published,
although the some results are scattered in the literature
\cite{sajat-medium,sajat-FP-model}, and they will be treated in \cite{kohei-charges}.

For example it is known that the first non-trivial local charge for the periodic $H$
is given by \cite{sajat-medium}
\begin{equation}
  \sum_j  [h_j,h_{j+1}+h_{j+2}]=2\sum h_j(h_{j+1}+h_{j+2}).
\end{equation}
From this charge we can read off the following null-vector of $H_A$:
\begin{equation}
\sum_{a=1}^M  (\ket{a,a+1}+\ket{a,a+2}).
\end{equation}
Note that this null-vector is actually the difference of a Dicke state and a constrained Dicke state, both of which are
null-vectors of $H_A$, as argued above.

Further null-vectors can be found by considering other charges of $H$ in the periodic case.

Finally we note that in this model we can also find crosscap states as exact eigenstates, using the results of
\cite{sajat-crosscap}. Let $\ket{B}$ be the standard Bell pair given by
\begin{equation}
  \ket{B}=\frac{1}{\sqrt{2}} (\ket{\uparrow\downarrow}-\ket{\downarrow\uparrow}),
\end{equation}
and let $\ket{B}_{jk}$ denote a Bell pair prepared on the pair of sites $j$ and $k$. Assuming that $M$ is even we can
prepare the crosscap state
\begin{equation}
  \otimes_{j=1}^{M/2} \ket{B}_{j,j+M/2}.
\end{equation}
It follows from the results of \cite{sajat-crosscap} that this vector is a null-vector of $H_A$. This statement holds
generally for $H_A^{(N)}$ with even $N$. For odd values of $N$ other crosscap states can be found, but we do not pursue
this direction further here.

%\bibliography{pozsi-general}


\end{document}